\documentclass{article}
\pdfoutput=1

\usepackage{arxiv}

\usepackage[utf8]{inputenc} 
\usepackage[T1]{fontenc}    
\usepackage{hyperref}       
\usepackage{url}            
\usepackage{booktabs}       
\usepackage{amsfonts}       
\usepackage{nicefrac}       
\usepackage{microtype}      
\usepackage{graphicx}
\usepackage{natbib}
\usepackage{doi}

\title{Probable Approximate Coordination}

\date{} 					

\author{ \href{https://orcid.org/0000-0003-2721-5928}{\includegraphics[scale=0.06]{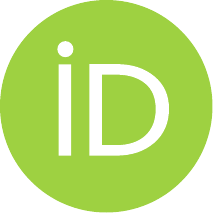}\hspace{1mm}Ariel Livshits}\ \\
	Yahoo! Research \\
	\texttt{livshits.ariel@gmail.com} \\
	\And
	\href{https://orcid.org/0000-0001-5549-1781}{\includegraphics[scale=0.06]{orcid.pdf}Yoram Moses} \\
	Technion, Israel \\
	\texttt{moses@technion.ac.il} \\
}

\date{}

\renewcommand{\headeright}{Extended Version}
\renewcommand{\undertitle}{This is the extended version of a paper accepted to the \\ 2023 Conference on Principles of Distributed Systems (OPODIS)}
\renewcommand{\shorttitle}{Probable Approximate Coordination}

\usepackage{amssymb,amsthm,amsmath}
\usepackage{xcolor,paralist,titlesec,fancyhdr,etoolbox}

\theoremstyle{plain}
\newtheorem{theorem}{Theorem}
\newtheorem{lemma}[theorem]{Lemma}
\newtheorem{corollary}[theorem]{Corollary}

\newtheorem{definition}[theorem]{Definition}

\theoremstyle{definition}

\theoremstyle{remark}

\newtheorem*{note*}{Note}

\newtheorem*{remark*}{Remark}

\newtheorem*{claim*}{Claim}

\usepackage[justification=centering]{caption}

\DeclareMathOperator{\E}{\mathbb{E}}
\DeclareMathOperator{\M}{\mathcal{M}}
\DeclareMathOperator{\T}{\mathcal{T}}
\DeclareMathOperator{\Oof}{\mathcal{O}}
\DeclareMathOperator{\D}{\mathcal{D}}
\DeclareMathOperator{\p}{\tilde{{\it p}}}

\usepackage{algorithm}
\usepackage{tikz}
\usepackage[noend]{algpseudocode}
\algdef{SE}[EVENT]{Event}{EndEvent}[1]{\textbf{upon receipt of a}\ #1\ \textbf{message,}  \algorithmicdo}{\algorithmicend\ \textbf{event}}%
\algtext*{EndEvent}
\algdef{SE}{InputEvent}{EndInputEvent}{\textbf{upon receiving an external input,} \algorithmicdo}{\algorithmicend\ \textbf{event}}%
\algtext*{EndInputEvent}
\newcommand{\brk}[1]{\left[ #1 \right]}
\newcommand{\brs}[1]{\left\{ #1 \right\}}
\newcommand{\prs}[1]{\left( #1 \right)}
\newcommand{\pr}[1]{\Pr\prs{#1}}
\newcommand{\abs}[1]{\left | #1 \right|}

\begin{document}
\maketitle

\begin{abstract}
We study the problem of how to coordinate the actions of independent agents in a distributed system where message arrival times are unbounded, but are determined by an exponential probability distribution. Asynchronous protocols executed in such a model are guaranteed to succeed with probability~1. We demonstrate a case in which the best asynchronous protocol can be improved on significantly. Specifically, we focus on the task of performing actions  by different agents in a linear temporal order---a problem known in the literature as \emph{Ordered Response}. In asynchronous systems, ensuring such an ordering requires the construction of a message chain that passes through each acting agent, in order. 
Solving \emph{Ordered Response} in this way in our model will  terminate in time that grows linearly in the number of participating agents $n$, in expectation. We show that relaxing the specification slightly allows for a significant saving in time. Namely, if \emph{Ordered Response} should be guaranteed with high probability (arbitrarily close to~1), it is  possible to significantly shorten the expected execution time of the protocol. We present two protocols that adhere to the relaxed specification.
One of our protocols executes exponentially faster than a message chain, when the number of participating agents $n$ is large, while the other is roughly quadratically faster. For small values of~$n$, it is also possible to achieve similar results by using a hybrid protocol. \\
\end{abstract} %%%%%%%%%

\centerline
{\small \bfseries \scshape Acknowledgements}
\begin{quote}
Yoram Moses is the Israel Pollak academic chair at the Technion. This work was supported in part by the Israel Science Foundation under grant 2061/19. The authors would like to thank Yonathan Shadmi for his most valuable help.
\end{quote}

\section{Introduction}

Numerous applications of distributed systems rely on promptly responding to spontaneously occurring events triggered by the environment. These events can range from the activation of fire alarms or smoke detectors to transactions involving bank account deposits or withdrawals. In order to ensure that the system behaves correctly, it may be necessary to execute one or more relevant actions. The sequential execution of these actions often holds significance, particularly when dealing with financial transactions. Typically, a diverse range of such actions, including condition testing and related updates, must be carried out to successfully complete these transactions. 

This paper considers solutions to a distributed coordination problem called \emph{Ordered Response} (\emph{OR}).
In \emph{Ordered Response}, agents must perform individual actions in a particular linear order, in response to the spontaneous arrival of an external input to the system. For example, let there be a system with three agents $i_1,i_2$ and $i_3$. Let $t_1, t_2$ and $t_3$ be the points in time at which the agents perform their respective actions. Then it must hold that $t_1 \leq t_2 \leq t_3 < \infty$ in every run of a protocol that solves \emph{OR}. Notice that if the agents act simultaneously, then the specifications of \emph{OR} are satisfied. 

The \emph{Ordered Response} problem has been analyzed previously in the contexts of  synchronous and asynchronous message-passing communication models \cite{ben2010beyond,chandy1985processes}. However, in many modern day multi-agent systems, message delays are subject to probabilistic bounds, i.e., determined according to a probabilistic distribution. Examples can be seen in wireless sensor networks (WSNs). A WSN is a distributed multi-agent system that is often  composed of a large quantity of disposable sensors that communicate using wireless broadcast. Messages sent in this network experience random delays due to random changes in the environment. However, it was observed in \cite{abdel2002analysis} that a single server M/M/1 queue \cite{kleinrock1975theory} can represent the cumulative link delay in a WSN, when the random delays are modeled as exponential random variables.

In this work, we investigate the \emph{Ordered Response} problem in a model where message arrival times are unbounded, but are sampled from an exponential probability distribution.\footnote{For modelling WSN link delays, the most
widely used distributions are Gaussian, exponential, gamma and Weibull \cite{leon1994probability,papoulis1991probability}.
We opted for the exponential distribution for its simplicity in analyzing the model, as considering other distributions would introduce technical challenges beyond the scope of our work.} We are particularly interested in protocols that achieve correctness with high probability (w.h.p.) and terminate in a short expected time. 
It is unclear, based on initial observations, whether the previous solutions are the most optimal for this particular task.

One of the conclusions that can be drawn from previous work done by Chandy \& Misra \cite{chandy1985processes}
is that ensuring a linear ordering of actions at distinct sites of an asynchronous system necessitates the existence of a message chain among the coordinating agents. Namely, in order for agents to ensure that their respective actions are performed in a linear temporal order, a message chain visiting each of the sites in this order must be formed. 
Concretely, every asynchronous protocol that solves \emph{Ordered Response} must construct a message chain among the agents. Consequently, in an asynchronous 
system of $n$ agents, \emph{Ordered Response} can only be solved in time that is linear in $n$, since a chain of at least $n-1$ 
messages must be sent. 

 Since message delays  are unbounded, our probabilistic model appears to be closer to an asynchronous model, than a synchronous model. However, the assumption of exponentially distributed delays may be leveraged to significantly shorten the expected time of termination. Intuitively, since they naturally induce a distribution on the runs of a given protocol, it is possible to choose appropriate waiting times for the agents such that correct runs are more likely to happen. For example, say an agent broadcasts a message to $99$ other agents. Assuming that arrival times are independent and identically distributed exponential random variables with parameter $1 \ [sec^{-1}]$, then, with probability  $0.999$, all messages will arrive within roughly $11.52$ seconds. So, assuming that every agent has access to a global clock and that the broadcast happened at time $t=0$, after $12$ seconds all agents may act simultaneously, achieving Ordered Response with high probability. In contrast  to a message chain, where the expected response time would be at best $99$ seconds, a significant improvement.\footnote{We assume that agents can send only a constant amount of messages for reasons we will go into later on.}

The main contributions of this paper are:
\begin{itemize}
    \item A study of \emph{Ordered Response} protocols in the Exponentially Distributed Delay (EDD) model, which we are the first to formally define.
    \item Two tunable \emph{Ordered Response} protocols are presented. Both solve the problem with probability as close to $1$ as desirable. The first assumes the existence of a global clock and executes \textbf{exponentially} faster than a message chain, when the number of participating agents $n$ is large. The second does not use a global clock, and instead employs a procedure that ensures that, with high probability, the local clocks of the agents are \emph{probably approximately} synchronized. As a result, it executes slower than the first former protocol, but is still quadratically faster than the message chain protocol.
    \item For each of the above protocols, we also provide a hybrid version that combines the protocol with a message chain to efficiently achieve \emph{Ordered Response} w.h.p.\  for any $n$.
\end{itemize}

\section{Related Work}

The \emph{Ordered Response} problem was introduced by Ben-Zvi \& Moses in \cite{ben2010beyond}, and  was investigated there in the context of the Bounded Communication Model, in which message delays have (deterministic) upper bounds.  Given such bounds, the mere passage of time can provide information about the occurrence of events at remote sites, without the need for explicit confirmation. The authors extend Lamport's \emph{happened-before} relation \cite{lamport2019time} to define a causal structure called the ``centipede'' (a strict generalization of the ``message chain"), and show  that centipedes must exist in every execution in which linear ordering of actions is ensured. The concurrent Ordered Response protocols presented in this work implement a similar communication approach to the centipede by broadcasting  information. Our model does not assume deterministic upper bounds on message delays, but rather probabilistic ones. The result is a protocol that solves Ordered Response with high probability.

Assuming an exponential distribution on message arrival times is not new. In fact, such assumptions are frequently made in the context of Wireless Sensor Networks (WSNs). Specifically, for the problem of estimating and/or bounding clock skew of sensors in WSNs \cite{abdel2002analysis, chaudhari2009energy, jeske2005maximum, jeske2003estimation, leng2011low, zennaro2013network}. It was observed in \cite{abdel2002analysis} that a single server M/M/1 queue can represent the cumulative link delay in a WSN, when the random delays are modeled as exponential random variables. Moreover, the Minimum Link Delay algorithm, proposed in  \cite{abdel2002analysis}, which shows good performance, was derived independently in \cite{jeske2005maximum} assuming exponentially distributed network delays. Therefore, the assumption of an exponential distribution seems to be suited for modeling WSN link delays.

In the following, we will introduce a Concurrent Ordered Response protocol designed for an environment lacking a global clock. This protocol employs a procedure to achieve probable approximate synchronization of the agents' local clocks. Previous literature has proposed clock synchronization protocols in environments with probabilistic message delays \cite{arvind1994probabilistic,cristian1989probabilistic,olson1994probabilistic, palchaudhuri2003probabilistic}. These protocols, similar to our solutions, guarantee clock synchronization with probability as close to $1$ as desirable. However, they rely on agents being able to remotely read another agent's local clock multiple times to accomplish this. The \emph{novelty} of our approach lies in its minimal communication requirements.\footnote{Minimal communication is crucial in WSNs due to battery life concerns.} Agents broadcast messages only once (or twice in the hybrid variant) in order to achieve probable approximate synchronization. Along with the subsequent second phase, the agents successfully achieve Ordered Response with high probability.

The rest of this work is organized as follows. Section \ref{preliminaries} presents our model and existing \emph{Ordered Response} solutions in the synchronous and asynchronous models. Section \ref{with_global_clock} presents the CORE protocol when a global clock is present in the model. As well as, a hybrid protocol of CORE with a message chain. Section \ref{without_global_clock} details a procedure to \emph{probably approximately} synchronize local clocks of agents, and how we use it in the CORE protocol when a global clock is not present in the model.  Finally, Section \ref{conclusions} concludes this work with our conclusions.

\section{Preliminaries}\label{preliminaries}

\subsection{Model Formulation}

In the Ordered Response problem, we consider a set of $n+1$ agents $\mathbb{P} \triangleq \{i_0, i_1,\ldots i_n\}$. The agents are assumed to be connected via a complete communication network, and agents can communicate with one another by sending and receiving messages (including broadcasts) that contain any form of information. We assume that the value of~$n$ is available as an input to the protocols, and every agent is aware of its own ID, as well as the IDs of all the other agents in the system. The system is considered to be event-driven, i.e., agents change their state only when some local event occurs. Agents wait between the occurrence of these events, and are assumed inactive prior to receiving the first message or external input from the environment in an execution.

We assume that all agents can measure the passage of time accurately using local clocks. A clock can be used to set a timer that will cause the agent to activate at some point in the future. Hence, any local event experienced by an agent is either the receipt of a message or a timer signalling. All agents are assumed to be inactive up until the first local event takes place. When an agent is activated, it immediately performs some local calculation and/or action as dictated by a deterministic distributed protocol that is shared among all agents in the system. We assume the agents are well-behaved and reliable, i.e., agents do not crash or disobey the protocol. 

We designate agent $i_0$ to be a ``supervisor'' agent. At time $t=0$, the agent $i_0$ receives an external input from the environment that triggers the system into action. 
Each of the other ``worker'' agents~$i_k$ has a special action $\alpha_k$ unique to itself. 
When the external input is delivered, the supervisor begins the process by sending messages to a subset of agents. From then on, each of the agents, once active, coordinates with the others 
with the goal of performing their respective actions in a \emph{linear temporal order}. 

Formally, let $\mathcal{P}$ be some shared protocol, and denote the set of runs of  $\mathcal{P}$ by $R(\mathcal{P})$. We say that a run $r\in R(\mathcal{P})$ adheres to the specifications of \emph{Ordered Response} if the following two properties hold:
\begin{itemize}
    \item \emph{Safety:}\quad  $\forall  k > 1:$ agent $i_k$ does not perform the action $\alpha_k$ before agent $i_{k-1}$ performs the action $\alpha_{k-1}$,

    \item \emph{Liveness:}\quad  every agent $i_k$ eventually performs action $\alpha_k$.
\end{itemize}
Accordingly, we say that $\mathcal{P}$ solves \emph{Ordered Response} if every run $r\in R(\mathcal{P})$ of the protocol adheres to the \emph{Ordered Response} specifications. Later, we will investigate protocols that do not solve \emph{Ordered Response}, but rather solve it with high probability. We say that a protocol $\mathcal{P}$ solves \emph{Ordered Response} with some probability $p$, if the probability that a run of the protocol 
will satisfy both \emph{Safety} and \emph{Liveness} is at least~$p$.

\subsection{A Probability Measure on Runs}

Unlike in classical message-passing communication models, we assume that message delays are determined by a probability distribution. Specifically, the arrival times of messages that are sent over any communication channel, are assumed to be sampled from an exponential distribution with parameter $\lambda$. The delays of distinct messages (including those that are sent via broadcast) are assumed to be statistically independent of one another. We name this communication model the Exponentially Distributed Delay model, denoted by the abbreviation EDD, which is shorthand for $\text{EDD}(\lambda)$ when $\lambda$ is clear from the context.

In this work, we exclusively consider deterministic protocols, and so the communication infrastructure provides the only source of randomness. The exponential distribution on message delays induces a probability distribution on the set of runs of any given protocol $\mathcal{P}$. We map every run of a protocol $\mathcal{P}$ in the EDD model, to a sampling of a countably infinite set of exponential random variables with parameter $\lambda$. Every exponential random variable in this set represents the delay of a message sent in a run of the protocol. This induces a natural probability space $\mathcal{S} = (\Omega, \mathcal{F}, \Pr)$ that we use to prove our claims (see Appendix \ref{prob_space_form} for a complete exposition of the probability space).

In a probabilistic environment, a simplistic solution to the Ordered Response problem may involve repeatedly transmitting the same message between agents. The intention is to increase the likelihood of an earlier message arrival beyond what a single message's probability distribution allows. This approach increases 
the message complexity of the protocol and, perhaps more crucially, imposes considerable energy costs on the individual agents, which is often unreasonable, e.g., in WSNs. 
Therefore, we assume any agent can only transmit a constant number of messages (i.e., independent of the number of agents in the system). If a protocol does make use of a constant number of retransmissions of the same message to the same destination (say $k$ such messages), we can effectively model this as a single message that is sampled from an exponential distribution with a larger parameter ($\tilde{\lambda} \geq k \lambda$).

\subsection{Performance and Correctness Metrics}

In this work, we investigate protocols that ensure \emph{Ordered Response} with high probability. Specifically, we allow protocols that ensure a linear temporal order with some probability $p$. A run in which both the \emph{Safety} and \emph{Liveness} conditions hold we call a \emph{correct run}. Formally, let $t_k : \Omega \to \mathbb{R}$ be a random variable that represents the time that agent $i_k$ performs $\alpha_k$. The value of $t_k$ is determined by the protocol $\mathcal{P}$ and a realization $\omega \in \Omega$, i.e., for every agent $i_k$,  the protocol $\mathcal{P}$ induces a function $\tau_k$ such that $t_k = \tau_k(\omega)$. The probability of a correct run of $\mathcal{P}$ is then given by ${p \triangleq \Pr(t_1 \leq t_2 \leq \ldots \leq t_n < \infty)}$.

We compare protocols using two performance measurements that are a function of the number of agents $n$. The first is message complexity, as in the number of messages sent, and the second is \emph{response time}. The response time of an Ordered Response protocol in a given run is defined as the time at which the last agent performs its unique action. Thus, the response time of the protocol is a random variable $RT \triangleq \max\{t_1, \ldots, t_n\}$. Since message delays are unbounded in the EDD model, there exist runs of the protocol that never terminate. Although the probability measure of these runs is zero, measuring the response time of protocols in this model is meaningless. Instead, we measure the \emph{expected} response time, denoted by $\E[RT]$.

\subsection{Asynchronous Ordered Response and the Message Chain Protocol}

In an asynchronous system, there is no global clock and no bound on the arrival time of messages. Fortunately, there is a very simple Ordered Response protocol in this model as well, called the Message Chain protocol. At time $t=0$, the supervisor $i_0$ sends a message to~$i_1$. For all $k \geq 1$, when $i_k$ receives a message from $i_{k-1}$ it performs $\alpha_k$ and sends a message to agent $i_{k+1}$ (if $k < n$). This creates a message chain among agents. 

Since communication in the asynchronous setting is reliable, the message chain protocol ensures \emph{Safety} and \emph{Liveness} 
in every run. The correctness of the protocol stems from the fact that no agent acts before receiving a direct message from its predecessor in the agent ordering. Notice that a direct message is not a necessary condition for an agent to act. Rather, as the analysis in Chandy \& Misra \cite{chandy1985processes} suggests, a message chain between every two consecutive agents in the ordering is required to ensure a linear temporal ordering of actions. Therefore, the message chain protocol presented above is an optimal solution for the asynchronous model in terms of both message complexity and response time.

To measure the response time of a protocol in an asynchronous system, we consider each message arrival time to be of length of one time unit.\footnote{This is generally the accepted approach in the literature for measuring time of execution of a run in an asynchronous system \cite{lynch1996distributed}.} Under such accounting, the message chain protocol exhibits $\Theta(n)$ response time, due to its sequential behavior. Since the Message Chain protocol is the optimal Ordered Response protocol in an asynchronous system, this is also a lower bound for the asynchronous Ordered Response problem, in general.

\subsection{Ordered Response with Probability 1}

While message arrival times are not deterministically bounded in the EDD model,  they are bounded in the probabilistic sense. In particular, messages are more likely to arrive sooner rather than later. Our focus in this work will be the design and analysis of Ordered Response protocols that achieve good performance in the EDD model. 

We consider the Message Chain protocol as our baseline. Notice that the protocol ensures that the probability of a \emph{correct run} is $1$. The reason is that the protocol ensures \emph{Liveness} only \emph{w.p.}~$1$, due to the fact that messages in the EDD model arrive in finite time \emph{w.p.}~$1$.

\begin{definition}
We say that a run $r \in R(\mathcal{P})$ contains a message chain, if there exist a set of~$n$ messages $m_1, m_2, \ldots, m_n$ that are sent in $r$, and for all $1 \leq k \leq n-1$ the agent $i_k$ sends $m_{k+1}$ only after it receives $m_k$ from agent $i_{k-1}$.
\end{definition}

Additionally, in a run that contains a message chain, we say that an agent $i_k$ ``receives a message chain of length $k$ at time $t$'' if at time $t$ agent $i_k$ receives message $m_k$.

\begin{theorem}\label{expected_response_lemma}
In the EDD model, if every run of a protocol $\mathcal{P}$ contains a message chain, then it holds for $\mathcal{P}$ that $E[RT] \in \Omega(\frac{n}{\lambda})$.
\end{theorem}
\begin{proof}
If the run implements a message chain, then agent $i_{n-1}$ does not terminate prior to receiving a message chain of length $n-1$, since it has to send the final message $m_n$. By definition, the response time of a protocol in a given run is greater or equal to the time at which $i_{n-1}$ terminates. Thus, $RT > \xi$ where $\xi$ is a sum of $n-1$ \emph{i.i.d.} exponential random variables with parameter $\lambda$. Hence, $\E[RT] > \E[\xi] = \frac{n-1}{\lambda}$.
\end{proof}

From Theorem \ref{expected_response_lemma} we conclude that the expected response time of the Message Chain protocol in the EDD model is in $\Omega(\frac{n}{\lambda})$. Later on, we investigate protocols that trade off probability of correctness for expected response time that is sub-linear in the number of participating agents $n$.

\section{Ordering with a Global Clock}\label{with_global_clock}
In order to investigate the design of faster protocols, we relax the specification slightly to allow the protocol to guarantee \emph{Safety} and \emph{Liveness} properties with high probability (\emph{w.h.p.}). More precisely, we now consider protocols that solve Ordered Response with probability $0<p<1$, where $p$ can be as close to $1$ as desirable. Our goal is to find a protocol that significantly outperforms the Message Chain protocol, in terms of expected response time. 

The main disadvantage of the Message Chain protocol is its sequential behaviour. However, in the EDD model, messages are more likely to arrive earlier rather than later. Our approach leverages this towards speeding-up the response time of the protocol. The main idea behind the protocols that will be described in this section is the concurrent transmission of information to all parties. For example, instead of sending only one message to one agent informing him that the external input has arrived, the supervisor broadcasts this information to all agents in the system. We call this approach -- Concurrent Ordered REsponse, abbreviated by CORE.

\subsection{The CORE Protocol}

In a synchronous system, there usually exists a deterministic upper bound $\Delta$ on the arrival time of messages. In the CORE protocol there is a parameter $\D$ that is configurable, rather than being a constant defined by the environment. This parameter is chosen so that \emph{w.h.p.} all the messages that are sent in the supervisor's broadcast are delivered within time $\D$. The psuedocode of the protocol is presented in Algorithm \ref{synchronous_CORE}.

\begin{algorithm}
\caption{The CORE($\D$) Protocol}
\label{synchronous_CORE}
\begin{algorithmic}[1]
\Procedure{Protocol for Agent $i_0$}{}
    \InputEvent
        \State{send ``\textbf{trigger}" to all}
    \EndInputEvent
\EndProcedure
\newline
\Procedure{Protocol for Agent $i_k$}{}
\Event{``\textbf{trigger}"}
    \State{Wait until current time $\ge \D$}
    \State{perform $\alpha_k$}
\EndEvent
\EndProcedure
\end{algorithmic}
\end{algorithm}

At time $t=0$, the supervisor agent receives an external input from the environment. The agent then proceeds to broadcast a message to all other agents containing the word ``trigger". Agents that receive a ``trigger" message from the supervisor, wait until the global clock reads $t = \D$, and then act.%
\footnote{Notice that under our protocol several actions can be performed simultaneously, which is consistent with the original definition of Ordered Response in \cite{ben2010beyond}. But the CORE protocol can be easily modified at essentially no cost to support strict ordering of the actions. I.e., if~$t_i$ denotes the time at which the $i$'th action is performed, then $t_{k+1}> t_k$, rather than $t_{k+a}>t_k$, in good runs of the protocol. Namely, for an arbitrarily chosen small~$\epsilon$, we can modify line 6 of the protocol so that Agent~$i_k$ waits until time $\ge \D+\,(1-2^{-k})\cdot\epsilon$. In run {\it (c)} of Figure~\ref{sync_response_ex}, for example, the actions will be performed in strict linear temporal order, very soon after time~$\D$; indeed, they would all complete before time $\D+\epsilon$.} If  an agent receives the message after time $\D$, it does not wait and acts immediately. Several examples of runs of the CORE protocol are illustrated in Figure \ref{sync_response_ex}.

\usetikzlibrary{trees,positioning,shapes,calc}
\captionsetup{justification=raggedright,singlelinecheck=false}
\begin{figure}[h]
%\label{fig1}
\vspace{3em}
  %\hspace{-0.05\linewidth}
  \begin{tikzpicture}
      \pgfmathsetmacro{\xscale}{0.6pt}
      \pgfmathsetmacro{\yscale}{0.6pt}
      \pgfmathsetmacro{\cradius}{6pt}
      \begin{tikzpicture}[
        dot/.style = {circle, fill, minimum size=#1,
                      inner sep=0pt, outer sep=0pt},
        dot/.default = 6pt  % size of the circle diameter 
        ]  

      % Nodes:
      
      \node (tl) {};
      \node (tr) [right=8*\xscale of tl] {};
      
      \node (t0l) at ($(tl)!0.15!(tr)$) {};
      \node (tdl) at ($(tl)!0.6!(tr)$) {};
      \node (t0h) [above=6*\yscale of t0l] {};
      \node (tdh) [above=6*\yscale of tdl] {};
      
      \node (i0l) at (t0h-|tl) {};
      \node (i0r) at (tdh-|tr) {};
      
        % Dashed Gray Lines
      \node (i1l) at ($(i0l)!0.2!(tl)$) {};
      \node (i1r) at ($(i0r)!0.2!(tr)$) {};
      \node (i2l) at ($(i0l)!0.4!(tl)$) {};
      \node (i2r) at ($(i0r)!0.4!(tr)$) {};
      \node (i3l) at ($(i0l)!0.6!(tl)$) {};
      \node (i3r) at ($(i0r)!0.6!(tr)$) {};
      \node (i4l) at ($(i0l)!0.8!(tl)$) {};
      \node (i4r) at ($(i0r)!0.8!(tr)$) {};
      
      % Lines:
      \draw[thick,->] (tl.center) -- (tr.center) node[anchor=north] {time};
      \draw[thick,dashed] (t0h.center) -- (t0l.center) node[anchor=north] {$0$};
      \draw[thick,dashed] (tdh.center) -- (tdl.center) node[anchor=north] {$\D$};
      \draw[dashed,gray] (i0l) -- (i0r) node[anchor= north east] {$i_0$};
      \draw[dashed,gray] (i1l) -- (i1r) node[anchor= north east] {$i_1$};
      \draw[dashed,gray] (i2l) -- (i2r) node[anchor= north east] {$i_2$};
      \draw[dashed,gray] (i3l) -- (i3r) node[anchor= north east] {$i_3$};
      \draw[dashed,gray] (i4l) -- (i4r) node[anchor= north east] {$i_4$};
      
        % Circles
      \node[black,dot=\cradius] (trg) at (i0l-|t0l) {};
      \node[blue,dot=\cradius] (alpha1) at (i1r-|tdl) {};
      \node[blue,dot=\cradius] (alpha2) at (i2r-|tdl) {};
      \node[blue,dot=\cradius] (alpha3) at ($(i3l)!0.68!(i3r)$) {};
      \node[blue,dot=\cradius] (alpha4) at ($(i4l)!0.75!(i4r)$) {};
      
      % Arrows
      \draw[thick,->] (trg)--($(i1l)!0.5!(i1r)$);
      \draw[thick,->] (trg)--($(i2l)!0.3!(i2r)$);
      \draw[thick,->] (trg)--(alpha3.west);
      \draw[thick,->] (trg)--(alpha4.west);
      
      % Caption
      \node (cap_anch) at ($(tl)!0.5!(tr)$) {};
      \node (caption) [below=\yscale of cap_anch] {$(a)$ A correct run of CORE};
  \end{tikzpicture}
  \hfill
  \begin{tikzpicture}[
        dot/.style = {circle, fill, minimum size=#1,
                      inner sep=0pt, outer sep=0pt},
        dot/.default = 6pt  % size of the circle diameter 
        ]  
      % Nodes:
      
      \node (tl) {};
      \node (tr) [right=8*\xscale of tl] {};
      
      \node (t0l) at ($(tl)!0.15!(tr)$) {};
      \node (tdl) at ($(tl)!0.6!(tr)$) {};
      \node (t0h) [above=6*\yscale of t0l] {};
      \node (tdh) [above=6*\yscale of tdl] {};
      
      \node (i0l) at (t0h-|tl) {};
      \node (i0r) at (tdh-|tr) {};
      
        % Dashed Gray Lines
      \node (i1l) at ($(i0l)!0.2!(tl)$) {};
      \node (i1r) at ($(i0r)!0.2!(tr)$) {};
      \node (i2l) at ($(i0l)!0.4!(tl)$) {};
      \node (i2r) at ($(i0r)!0.4!(tr)$) {};
      \node (i3l) at ($(i0l)!0.6!(tl)$) {};
      \node (i3r) at ($(i0r)!0.6!(tr)$) {};
      \node (i4l) at ($(i0l)!0.8!(tl)$) {};
      \node (i4r) at ($(i0r)!0.8!(tr)$) {};
      
      % Lines:
      \draw[thick,->] (tl.center) -- (tr.center) node[anchor=north] {time};
      \draw[thick,dashed] (t0h.center) -- (t0l.center) node[anchor=north] {$0$};
      \draw[thick,dashed] (tdh.center) -- (tdl.center) node[anchor=north] {$\D$};
      \draw[dashed,gray] (i0l) -- (i0r) node[anchor= north east] {$i_0$};
      \draw[dashed,gray] (i1l) -- (i1r) node[anchor= north east] {$i_1$};
      \draw[dashed,gray] (i2l) -- (i2r) node[anchor= north east] {$i_2$};
      \draw[dashed,gray] (i3l) -- (i3r) node[anchor= north east] {$i_3$};
      \draw[dashed,gray] (i4l) -- (i4r) node[anchor= north east] {$i_4$};
      
        % Circles
      \node[black,dot=\cradius] (trg) at (i0l-|t0l) {};
      \node[blue,dot=\cradius] (alpha1) at (i1r-|tdl) {};
      \node[blue,dot=\cradius] (alpha2) at (i2r-|tdl) {};
      \node[red,dot=\cradius] (alpha3) at ($(i3l)!0.7!(i3r)$) {};
      \node[blue,dot=\cradius] (alpha4) at (i4r-|tdl) {};
      
      % Arrows
      \draw[thick,->] (trg)--($(i1l)!0.5!(i1r)$);
      \draw[thick,->] (trg)--($(i2l)!0.3!(i2r)$);
      \draw[thick,->] (trg)--(alpha3.west);
      \draw[thick,->] (trg)--($(i4l)!0.25!(i4r)$);

      % Caption
      \node (cap_anch) at ($(tl)!0.5!(tr)$) {};
      \node (caption) [below=\yscale of cap_anch] {$(b)$ An incorrect run of CORE};
  \end{tikzpicture}
  \hfill
  \begin{tikzpicture}[
        dot/.style = {circle, fill, minimum size=#1,
                      inner sep=0pt, outer sep=0pt},
        dot/.default = 6pt  % size of the circle diameter 
        ]  
      % Nodes:
      
      \node (tl) {};
      \node (tr) [right=8*\xscale of tl] {};
      
      \node (t0l) at ($(tl)!0.15!(tr)$) {};
      \node (tdl) at ($(tl)!0.8!(tr)$) {};
      \node (t0h) [above=6*\yscale of t0l] {};
      \node (tdh) [above=6*\yscale of tdl] {};
      
      \node (i0l) at (t0h-|tl) {};
      \node (i0r) at (tdh-|tr) {};
      
        % Dashed Gray Lines
      \node (i1l) at ($(i0l)!0.2!(tl)$) {};
      \node (i1r) at ($(i0r)!0.2!(tr)$) {};
      \node (i2l) at ($(i0l)!0.4!(tl)$) {};
      \node (i2r) at ($(i0r)!0.4!(tr)$) {};
      \node (i3l) at ($(i0l)!0.6!(tl)$) {};
      \node (i3r) at ($(i0r)!0.6!(tr)$) {};
      \node (i4l) at ($(i0l)!0.8!(tl)$) {};
      \node (i4r) at ($(i0r)!0.8!(tr)$) {};
      
      % Lines:
      \draw[thick,->] (tl.center) -- (tr.center) node[anchor=north] {time};
      \draw[thick,dashed] (t0h.center) -- (t0l.center) node[anchor=north] {$0$};
      \draw[thick,dashed] (tdh.center) -- (tdl.center) node[anchor=north] {$\D$};
      \draw[dashed,gray] (i0l) -- (i0r) node[anchor= north east] {$i_0$};
      \draw[dashed,gray] (i1l) -- (i1r) node[anchor= north east] {$i_1$};
      \draw[dashed,gray] (i2l) -- (i2r) node[anchor= north east] {$i_2$};
      \draw[dashed,gray] (i3l) -- (i3r) node[anchor= north east] {$i_3$};
      \draw[dashed,gray] (i4l) -- (i4r) node[anchor= north east] {$i_4$};
      
        % Circles
      \node[black,dot=\cradius] (trg) at (i0l-|t0l) {};
      \node[blue,dot=\cradius] (alpha1) at (i1r-|tdl) {};
      \node[blue,dot=\cradius] (alpha2) at (i2r-|tdl) {};
      \node[blue,dot=\cradius] (alpha3) at (i3r-|tdl) {};
      \node[blue,dot=\cradius] (alpha4) at (i4r-|tdl) {};
      
      % Arrows
      \draw[thick,->] (trg)--($(i1l)!0.5!(i1r)$);
      \draw[thick,->] (trg)--($(i2l)!0.3!(i2r)$);
      \draw[thick,->] (trg)--($(i3l)!0.7!(i3r)$);
      \draw[thick,->] (trg)--($(i4l)!0.25!(i4r)$);
      
      % Caption
      \node (cap_anch) at ($(tl)!0.5!(tr)$) {};
      \node (caption) [below=\yscale of cap_anch, align=center] {$(c)$ Same run, with a larger $\D$};
  \end{tikzpicture}
  \end{tikzpicture}
  \vspace{-4em}
\caption{ Three examples of runs of the CORE$(\D)$ protocol. \\ 
At time $t=0$, agent $i_0$ broadcasts a message to all other agents. The blue and red dots mark the points in time when the agents perform their actions. 
$(a)$ -- Agents $i_3$ and $i_4$ receive their messages late and act immediately and in order. This is a lucky run of the protocol.
$(b)$ -- Agents $i_3$ receives its message late and acts immediately, but out of order. This is an unlucky run of the protocol.
$(c)$ -- Same as the previous run, but $\D$ is chosen to be larger. All agents receive their messages before time $\D$ and act simultaneously.}

\label{sync_response_ex}
\end{figure}
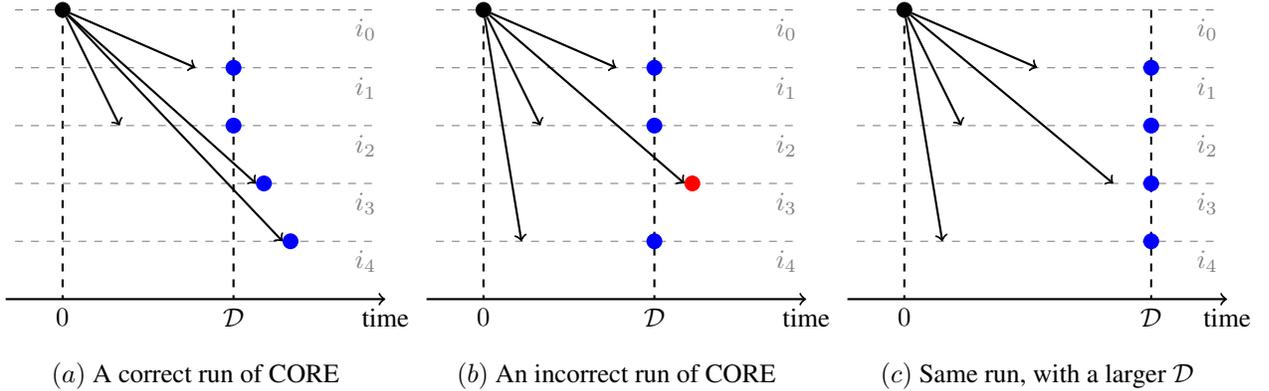

\subsection{Probability of a Correct Run}

If all messages arrive within time $\D$, then the run will be correct. However, if messages arrive late, this does not mean the run is ruined. For example, suppose the message sent to agent $i_k$ arrives late. Then there is some positive probability that all messages to agents $i_{k+1}, \ldots, i_{n}$ also arrive late, and that they will arrive in a linear temporal order.

\begin{theorem}\label{thm_corr_prob}
Fix $\D > 0$, and denote $q \triangleq 1-e^{-\lambda \D}$, then the probability of a correct run of the CORE$(\D)$ protocol is $ \sum_{k=0}^{n} \frac{1}{k!}q^{n-k}(1-q)^k$.
\end{theorem}

See Appendix \ref{thm_corr_prob_proof} for proof. As can be seen from Theorem \ref{thm_corr_prob}, the probability of a correct run depends both on $\D$ and $n$. Also, it includes runs that are ``lucky", similarly to the one presented in Figure \ref{sync_response_ex} $(a)$. The probability that a run is correct ``by design", i.e., every message arrives by time $\D$, is equal to $q^n = (1-e^{-\lambda \D})^n$ and, in particular, the probability of a correct run tends to $1$ as $\D$ tends to $\infty$, as can be seen from the following Corollary:

\begin{corollary}\label{cor_corr_prob}
The probability of a correct run of the CORE$(\D)$ protocol is at least \newline $(1-e^{-\lambda \D})^n$.
\end{corollary}

\begin{proof}
From Theorem \ref{thm_corr_prob}, the probability of a correct run of CORE$(\D)$ is given by:
\begin{equation}
    \sum_{k=0}^{n} \frac{1}{k!}q^{n-k}(1-q)^k = q^n + \sum_{k=1}^{n} \frac{1}{k!}q^{n-k}(1-q)^k \geq q^n = (1-e^{-\lambda \D})^n
\end{equation}
\end{proof}
Corollary \ref{cor_corr_prob} characterizes the relationship between the probabilistic communication bound $\D$ and the probability of a correct run of the CORE protocol. We now derive a closed-form expression for the probabilistic communication bound $\D$.

\begin{corollary}\label{Delta}
Fix $0<p<1$, and let $\D \geq \frac{-\ln(1-\sqrt[n]{p})}{\lambda}$. Then the probability of a correct run of the CORE$(\D)$ protocol in a setting with a global clock, is at least $p$.
\end{corollary}

\begin{proof}
Follows directly from Corollary \ref{cor_corr_prob}:
$ \quad \D \geq \frac{-\ln(1-\sqrt[n]{p})}{\lambda} \Rightarrow p \leq (1-e^{-\lambda \D})^n$
\end{proof}

\subsection{Expected Response Time}

We are interested in the expected response time of the CORE protocol. Specifically, this is the expected time of the last agent to perform its action. In a given run of the CORE protocol, the response time is determined by the arrival times of the messages from the supervisor's broadcast. They either all arrive by time $\D$, or at least one of them is delayed longer. 

Let $e_1, e_2,\ldots,e_n$ be the set of exponential random variables that are associated with messages sent in the supervisor's broadcast. Let $\M = \max\{e_1, e_2, \ldots, e_n\}$ be a random variable that denotes the time when the last message in a run is delivered. Then the response time of a run of the protocol is given by $RT = \max\{\D, \M\}$.

We are interested in the expected value of the response time $\E[RT]$. However, deriving the expectation of the above non-linear function of a random variable and a constant is a non-trivial exercise. Fortunately, there is a simple bound on $\E[RT]$ that is good enough for our purposes:
\begin{equation}\label{ert_bound}
    \E[RT] = \E[\max\{\D, \M\}] \leq  \D + \E[\M]
\end{equation}
The inequality follows from the fact the maximum is smaller than the sum, and the linearity of probabilistic expectation. Before we proceed in characterizing the expected response time, we prove a useful lemma:
\begin{lemma}\label{lemma_h_n}
$E[\M] = \frac{H_n}{\lambda}$, where $H_n$ is the $n^{\text{th}}$ harmonic number; i.e. $H_n = \sum_{m=1}^n \frac{1}{m}$.
\end{lemma}

See Appendix \ref{lemma_h_n_proof} for proof. The harmonic numbers roughly approximate the natural logarithm function \cite{conway1998book}, i.e., $H_n \in \Theta(log(n))$. As a result, we gain the following theorem, whose proof can be found in Appendix \ref{sync_core_ert_proof}:

\begin{theorem}\label{sync_core_ert}
Fix $0<p<1$, and let $\D \geq \frac{-\ln(1-\sqrt[n]{p})}{\lambda}$. 
The expected response time of the CORE$(\D)$ protocol is logarithmic in the number of participating agents i.e., $\E[RT] \in \Oof\ \left(\frac{\log(n)}{\lambda}\right)$.
\end{theorem}

\subsection{The CORE-Message-Chain Hybrid Protocol}

From Theorem \ref{sync_core_ert}, we see that the CORE protocol's expected response time grows logarithmically in $n$. In comparison, the Message Chain protocol's expected response time grows linearly in $n$. Hence, the CORE protocol is exponentially faster than the Message Chain protocol. However, the Message Chain protocol, while slower for large $n$, guarantees Ordered Response \emph{w.p.} $1$. Nonetheless, a case can be made in favour of the CORE protocol if it can maintain its edge for any $n$, and for any probability of a correct run $p$ that is as close to $1$ as desirable. 

Currently, this is not the case. Recall that $\D$ tends to $\infty$ as $p$ tends to $1$. As a result, for any given $n$, there exists some probability $\bar{p}$ such that for all $\bar{p} < p < 1$ the Message Chain protocol's expected response time $\frac{n}{\lambda}$ is smaller than $\D = \frac{-\ln(1-\sqrt[n]{p})}{\lambda}$. The key insight here is that for small $n$ the Message Chain protocol is very efficient and quick, but for large $n$ it suffers from a significant slowdown in performance. Our solution is to combine the two protocols. Essentially, we now design a protocol that runs both protocols concurrently, in order to benefit from the best of both worlds. 

In this hybrid protocol, the supervisor broadcasts the ``trigger'' message to all agents, and also sends a special message containing the word ``act'' to $i_1$. When $i_1$ receives this message, it starts a message chain that will pass through all agents. Agents that receive a message chain perform their action immediately, pass it on to the next agent in the ordering, and then terminate. Until then, they behave according to the CORE protocol described in Algorithm \ref{synchronous_CORE}.

\subsection{Performance and Probability of Correctness}

Deriving an exact expression for the probability of a correct run in the hybrid protocol is difficult. This is mainly due to the fact that in the hybrid protocol an agent may perform its action for two reasons. The agent may act either because it received a message chain, or because it previously received a message from the supervisor and the the global time is $\D$. It is enough for our purposes to lower bound the probability of a correct run. 
Similarly to the pure CORE protocol, if all messages from the trigger broadcast are delivered by time $\D$, then Ordered Response is guaranteed. From Corollary \ref{cor_corr_prob} we have that:
\begin{theorem}\label{hybrid_corr_prob}
Fix a parameter $\D > 0$. The probability of a correct run in the CORE$(\D)$-Message-Chain protocol is at least $(1-e^{\lambda \D})^n$.
\end{theorem}

We focus the reader's attention on two important observations. The first is that, in terms of expected response time, the hybrid protocol is no worse than the Message Chain protocol in every run. Once an agent receives a message chain, the agent immediately performs its action and terminates. The second observation is that, from some value of $n$, the expected response time stops growing linearly in $n$, and starts growing logarithmically in $n$. Thus, it holds for the hybrid protocol that $\E[RT] \in  \Oof\ \left(\frac{\log(n)}{\lambda}\right)$.

\begin{figure}
    \centering
    \includegraphics[width=0.7\textwidth]{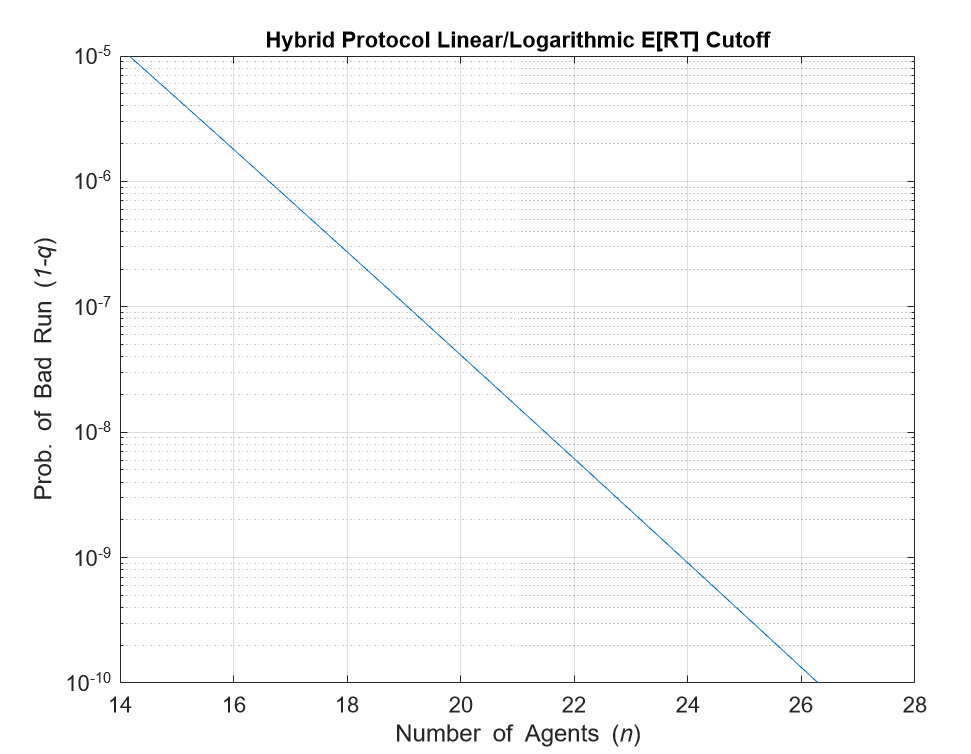}
    \captionsetup{justification=centering}
    \caption{Plot of the Linear/Logarithmic expected response time of the CORE-Message-Chain protocol.}
    \vspace{-1em}
    \label{fig_cutoff}
\end{figure}

With high probability, the message chain is automatically truncated once time $\D$ is reached, since no agent propagates the message chain after termination. For small $n$, and large $\D$, the message chain will more likely terminate the agents before time $\D$, which yields performance linear in $n$. However, when $n$ is large, the message chain will be much slower. In this case, the CORE part of the hybrid protocol dominates, and the expected response time will be logarithmic in $n$.

In Figure \ref{fig_cutoff}, given some probability of an incorrect run $1-p$, we plot for which $n$ it holds that $\frac{n}{\lambda} = \D$. The plot illustrates the linear/logarithmic cutoff point of the hybrid protocol, i.e., for a given probability of a correct run $p<1$, what is the minimum number of participating agents $N$ required such that for all $n>N$ the expected response time stops growing linearly in $n$, and starts growing logarithmically in $n$.

Figure \ref{fig_cutoff} shows that the hybrid protocol exhibits linear performance only when $n$ is small. For example, if we allow the rate of incorrect runs to be $1$ in $10^6$ (i.e., $p = 1-10^{-6}$), then a lower bound on the number of participating agents needed for the hybrid protocol to exhibit logarithmic performance is roughly $n=18$. For $p = 1-10^{-9}$, the lower bound is $n = 24$. As can be seen from the slope of the graph, the hybrid protocol requires an additional $3$ participating agents for a tenfold improvement in the probability of correctness. Hence, correctness probability of $p = 1 - 10^{-100}$, would require roughly $n = 300$ agents.
\newline

\section{CORE without a Global Clock}\label{without_global_clock}

In Section \ref{with_global_clock}, we presented a concurrent Ordered Response protocol for the EDD model when we can assume the existence of an accurate global clock. However, agents are only assumed to be able to measure time accurately using their own local clocks, though they may be initially offset relative to each other by some unknown quantity. 

We now consider the case where agents do not share a global clock, and so they initially have no idea how early or late they are relative to other agents. Previously, we assumed that the environment's external input would arrive at a point in time that we could denote as ``$t=0$", and that when an agent becomes active it would have access to the current global time. However, in the current model, the arrival of the initial external input is only observed by the supervisor, and the knowledge of when it arrived cannot be disseminated accurately to all the other agents.

This is a major challenge to the design of a concurrent Ordered Response protocol. To overcome this, we again leverage the probability distribution on message delays, to \emph{probably approximately} synchronize the local clocks of the agents. Specifically, the worker agents will attempt to approximate when the external input has arrived based on the arrival times of messages broadcast by the supervisor and synchronize their local clocks to fit this hypothesis. 

\subsection{Probably Approximately Synchronized Local Clocks}\label{sync_local_clocks_sec}

For ease of exposition, from this point onward assume that in addition to the supervisor agent $i_0$, there are $n+1$ worker agents named $i_1, i_2, \ldots, i_{n+1}$. As before, the supervisor agent $i_0$ broadcasts a ``trigger" message to all the other agents. When a worker agent receives a trigger message it rebroadcasts it to all of its coworkers, with the message ``redirect". Concurrently, it waits for $n$ such ``redirect" messages, and logs a timestamp of the arrival time of each such message according to its local clock. When the final message arrives, the agent proceeds to calculate its hypothesis of when the external input arrived based on the observed timestamps.

Agent $i_k$'s hypothesis $T_k$ is the mean of the measured timestamps $\{\tau_1, \tau_2, \ldots, \tau_n\}$ minus the expected delay of a message chain of length $2$, i.e. $T_k = \frac{1}{n}\sum_{m=1}^n \tau_m -\frac{2}{\lambda}$. Notice that the sequence of timestamps are \emph{i.i.d.} Erlang\footnote{An Erlang random variable is the probabilistic distribution of a sum of \emph{i.i.d.} exponential random variables (see \cite{erlang_dist}).} random variables with parameters $2$ and $\lambda$. By the Law of Large Numbers, as $n \to \infty$ their mean converges to their expected value: $\frac{2}{\lambda}$. This means that $T_k \xrightarrow[]{n \to \infty} 0$, or, in this case, it converges to the true arrival time of the external input to the system. 

The value $T_k$ is used by agent $i_k$ to offset its local clock. Formally, we denote by $C_{i_k}(t)$ the time that the local clock of agent $i_k$ shows at time $t$. We assume that $C_{i_0}(t) = t$, and that the external input arrives at $t=0$. When an agent $i_k$ hypothesizes that the external input arrived at time $T_k$, it sets an adjusted local clock $C_{i_k}^{Adj}(t)$ to be the value of $C_{i_k}(t)$ offset by $T_k$. I.e., for all $t>0$, agent $i_k$'s adjusted local clock shows that $C_{i_k}^{Adj}(t) = C_{i_k}(t) - T_k$.

Since the number of agents in the system is finite, the agents' local clocks are still offset relative to the supervisor's local clock. However, when the aforementioned procedure concludes, the relative offset between agents' local clocks can be bounded with some probability.

\begin{theorem}\label{delta_bound_thm}
Let $\tau_m \sim Erlang(2,\lambda)$ for $m=1,\ldots,n$, and let $T_k = \frac{1}{n}\sum_{m=1}^n \tau_m - \frac{2}{\lambda}$. Then:
\begin{equation}
   \pr{\max_{k \geq 1} \abs{T_k} \leq \frac{\ln(n)}{\lambda \sqrt{n} }} \geq \Psi(n)
\end{equation}
where:
\begin{equation}
\Psi(n) \triangleq \max \brs{0, \ 1 - \prs{1-\frac{\ln(n)}{2\sqrt{n}}}^{2n}n^{1+\sqrt{n}} - \prs{1+\frac{\ln(n)}{2\sqrt{n}}}^{2n}n^{1-\sqrt{n}}}.
\end{equation} 
In particular, $\Psi(n) \xrightarrow{n \to \infty} 1 $.
\end{theorem}

 See Appendix \ref{delta_bound_thm_proof} for proof. Denote $\delta \triangleq \frac{\ln(n)}{\lambda \sqrt{n}}$. Then by Theorem \ref{delta_bound_thm}, with probability at least $\Psi(n)$, the local clocks of all agents are \emph{$\delta$-synchronized}: 

\begin{definition}
If for all $k,m\geq 1$ it holds that $\abs{C_{i_k}^{Adj}(t) - C_{i_m}^{Adj}(t)} \leq \delta$, we say that the agents' local clocks are \emph{$\delta$-synchronized}. 
\end{definition}

As can be seen from Theorem \ref{delta_bound_thm}, the probability that the agents' local clocks will be $\delta$-synchronized converges to $1$ as $n \to \infty$. However, for $n < 115$, ${\Psi(n)} = 0$. This is due to the fact that in the proof of Theorem \ref{delta_bound_thm}, we use the Union Bound Theorem, which results in an overly course lower bound. As a result, for these values of $n$, the probabilistic bound given in Theorem \ref{delta_bound_thm} is not very informative. To see the actual probability of $\delta$-synchronization for small $n$, we have performed a small experiment. We sampled the vector $(T_1, T_2, \ldots, T_n)$ a thousand times and calculated the empirical mean over the binary results of whether the vector elements are at most $2\delta$ far from each other. As can be seen in Figure \ref{fig_prob_delta_sync}, the real probability of $\delta$-synchronization is much higher, and for $n > 400$ is nearly $1$. 

\subsection{The PA-CORE Protocol}

In a setting without a global clock, we design a two phase protocol that we call PA-CORE for \emph{Probable Approximate} Concurrent Ordered Response. The first phase is the $\delta$-synchronization of the agents' local clocks using the procedure described in Section \ref{sync_local_clocks_sec}. In the second phase, the agents use this to perform their actions in a linear temporal order. The pusdeocode of the protocol is presented in Algorithm \ref{asynchronous_CORE}. 

\begin{algorithm}
\caption{The PA-CORE($\D$) Protocol}
\label{asynchronous_CORE}
\begin{algorithmic}[1]
\Procedure{Protocol for Agent $i_0$}{}
    \InputEvent
        \State{send ``\textbf{trigger}" to all}
    \EndInputEvent
\EndProcedure
\newline
\Procedure{Protocol for Agent $i_k$}{}
    \State{\textbf{Setup}:}
    \State{\quad $m \leftarrow 0$}
    \newline
    \Event{``\textbf{trigger}"}
        \State{send ``\textbf{redirect}" to all}
    \EndEvent
    \newline
    \Event{``\textbf{redirect}"}
        \State{$m \leftarrow m+1$}
        \State{$\tau_m \leftarrow C_{i_k}(t)$} \Comment{Measure and record current time.}
        \If{$m = n$}
            \State{$T \leftarrow \frac{1}{n}\sum_{m=1}^n \tau_m -\frac{2}{\lambda}$}
            \State{$C_{i_k}^{Adj}(t) \leftarrow C_{i_k}(t) - T$}
            \Comment{Set adjusted local clock.}
        \State{Wait until $C_{i_k}^{Adj}(t) \geq \D +2\delta\cdot (k-1)$}
        \State{perform $\alpha_k$}
        \EndIf
    \EndEvent
\EndProcedure
\end{algorithmic}
\end{algorithm}

At time $t=0$, agent $i_0$ broadcasts a ``trigger" message to all other agents. When a worker agent $i_k$ receives a ``trigger" message, it broadcasts a message with the word ``redirect" to all worker agents. Upon receiving a ``redirect" message, the agent records a timestamp according to its local clock. After $n$ such ``redirect" messages have been received, the agent approximates the arrival time of the external input and offsets its local clock to match this hypothesis. 

\begin{figure}[t]
\centering
\includegraphics[width=0.7\textwidth]{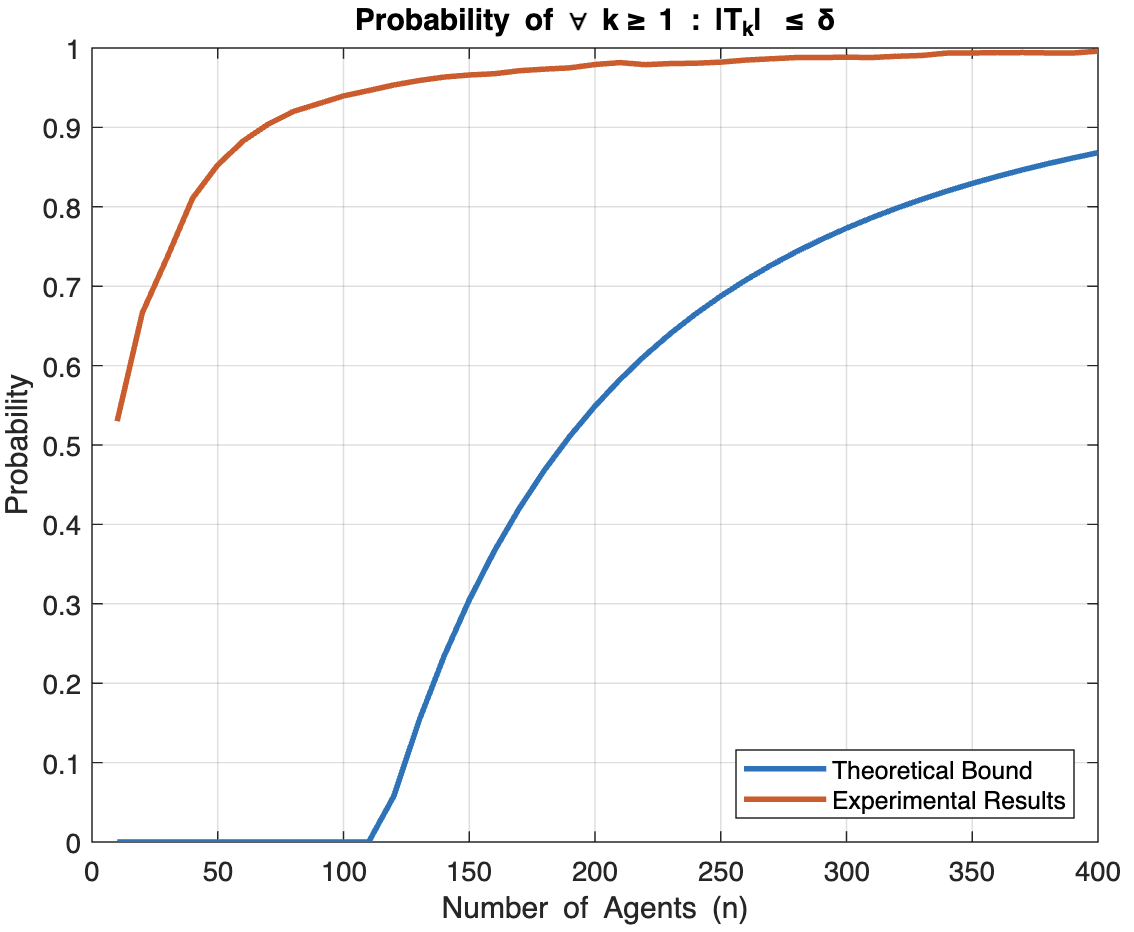}
\captionsetup{justification=centering}
\caption{Probability that clocks are $\delta$-synchronized. }
\label{fig_prob_delta_sync}
\vspace{-1.5em}
\end{figure}

Now, the agent $i_k$ moves into phase $2$. Being the $k$-th agent in the order, the agent waits until time $\D + 2\delta\cdot  (k-1)$, and then performs its action. If the local agent's time is already past this point when it receives the last ``redirect" message, it performs its action immediately. The additional delay of $2\delta\cdot (k-1)$ for agent $i_k$ ensures that agents act in non-overlapping time windows.

Notice that if $n$ was infinite, then $\delta = 0$, and the PA-CORE protocol would closely mirror the CORE protocol in Section \ref{with_global_clock}. Intuitively, the first phase of the protocol attempts to $\delta$-synchronize the agents' local clocks with high probability. However, $\delta$-synchronization is not a global clock, and so agents cannot act simultaneously. Instead, in phase 2 of the protocol, the agents wait an index-dependent time delay that guarantees that they act in non-overlapping windows. 

In contrast to the Message Chain protocol, the expected delay between the acting times of two consecutive agents is not a constant. In PA-CORE, this delay is at most $2\delta$ \emph{w.h.p.} (assuming that the agents' local clocks are $\delta$-synchronized). Recall that $\delta \triangleq \frac{\ln(n)}{\lambda \sqrt{n}}$. This value was chosen because the product $n\delta$ is in $\Oof\prs{\frac{\sqrt{n}\log(n)}{\lambda}}$ which we will show results in the expected response time of the protocol to grow sub-linearly in $n$.

\subsection{Probability of a Correct Run}

\begin{lemma}\label{prob_corr_given_lemma}
In a run of the PA-CORE$(\D)$ protocol, if all of the messages arrive by time $\D$ and the agents' local clocks are $\delta$-synchronized, then the run is guaranteed to be correct. 
\end{lemma}

\begin{proof}
We prove the claim by induction on the number of worker agents $n+1$.

\underline{Base Case}: ($n=1$) Since all messages arrive by time $\D$, and the agents' local clocks are $\delta$-synchronized, we have that $C_{i_1}^{Adj}(\D) \leq \D + \delta$. Therefore, agent $i_1$ acts no later than time $\D + \delta$. Also, for agent $i_2$ we have that $C_{i_2}^{Adj}(\D) \geq \D - \delta$. Since agent $i_2$ waits an additional $2\delta$ before acting, agent $i_2$ acts no earlier than time $\D + \delta$. Thus, agents $i_1$ and $i_2$ act in the proper linear temporal order.
\underline{Step}: Suppose that the claim holds for all $k \leq n$. By the same logic as in the base case, agent $i_{n+1}$ acts only after agent $i_n$, and by the induction hypothesis, agent $i_{n}$ also acts only after all the previous agents in the ordering have acted. 
\end{proof}

\begin{theorem}\label{prob_corr_thm}
Fix $\D > 0$. The probability of a correct run of the PA-CORE$(\D)$ protocol is at least $\Psi(n)\cdot (1-e^{-0.5\lambda \D})^{(n+1)^2}$.
\end{theorem}

See Appendix \ref{prob_corr_thm_proof} for proof.
Denote by $\p$ the probability that all messages arrive by time $\D$. By Theorem \ref{prob_corr_thm}, $\p = (1-e^{- 0.5\lambda \D})^{(n+1)^2}$, and the probability of a correct run is lower bounded by $\Psi(n)\cdot \p$, where $\Psi(n)$ is as defined in Theorem \ref{delta_bound_thm}. The probability $\p$ depends on the value of $\D$, and as $\D$ is increased, $\p$ draws closer to $1$. 

The lower bound we have derived is a product of two probabilities, $\Psi(n)$ and $\p$. While the latter is controllable by $\D$, the former is static, since it is solely dependent on $n$. From this point onward, we focus our analysis for large $n$, and we assume that $\Psi(n) \approx 1$, as the experimental results that are illustrated in Figure \ref{fig_prob_delta_sync}, suggest. We leave the proof of a tighter lower bound for future work.

\begin{corollary}\label{delta_tp_cor}
Fix $0<p<1$, and let $\D \geq \frac{-2\ln(1-\sqrt[(n+1)^2]{p})}{\lambda}$. Then, the probability of a correct run of the PA-CORE$(\D)$ protocol is at least $p$.
\end{corollary}
\begin{proof}
From Theorem \ref{prob_corr_thm}, the probability of a correct run is lower bounded by \\ ${(1-e^{-0.5\lambda \D})^{(n+1)^2}}$. Thus: $\quad  \D \geq -\frac{2\ln\prs{1-\sqrt[(n+1)^2]{p}}}{\lambda} \Rightarrow p \leq (1-e^{-0.5\lambda \D})^{(n+1)^2}$
\end{proof}
\subsection{Expected Response Times}

Let $RT_k$ be the response time of agent $i_k$ in a run of the protocol. Before an agent performs its action, it waits for $n$ message chains of length $2$, which it uses to approximate the external input arrival time. Independently of that, the agent also relays a message from the supervisor's original broadcast. Hence, the agent's response time is the maximum over two random variables.

Formally, let $e_k$ be the exponential random variables associated with the delivery time of the ``trigger" message sent to agent $i_k$, and let $\tau_1^k, \tau_2^k, \ldots, \tau_n^k$ be the timestamps measured by the agent when its receives the ``redirect" messages. Notice that $\{\tau_m^k\}_{m=1}^n$ is a set of \emph{i.i.d.} Erlang random variables with parameters $2$ and $\lambda$. Let $\T^k \triangleq \max\{\tau_1^k, \tau_2^k, \ldots, \tau_n^k\}$, then:
\begin{align}
    RT_k &= \max\brs{e_k, \max\brs{\T^k, \D + 2\delta\cdot (k-1)}} \\ &= \max\brs{e_k, \T^k, \D + 2\delta\cdot (k-1)} \leq \max\brs{e_k, \T^k, \D + 2\delta n}
\end{align}
Then, the expected response time can be upper bounded:
\begin{align}
    RT &= \max_k \brs{RT_k} = \max_k \brs{\max\brs{e_k, \T^k, \D + 2\delta n}}
    = \max\brs{ \max_k\brs{e_k, \T^k},  \D + 2\delta n} 
 \\ &\leq \max_k\brs{e_k, \T^k} + \D + 2\delta n  \quad \Rightarrow \quad \E[RT] \leq \E\brk{\max_k\brs{e_k, \T^k}} + \D + 2\delta n
\end{align}

\begin{theorem}\label{async_core_ert_thm}
Fix $0<p<1$, and let $\D \geq \frac{-\ln(1-\sqrt[n]{p})}{\lambda}$. 
The expected response time of the PA-CORE$(\D)$ protocol is sub-linear in the number of participating agents. In particular, $\E[RT] \in \Oof\prs{\frac{\sqrt{n}\log(n)}{\lambda}}$.
\end{theorem}

See Appendix \ref{async_core_ert_thm_proof} for proof.
In conclusion, we have presented a concurrent Ordered Response protocol that guarantees a correct run with high probability in a sub-linear expected response time. However, the trade-off is increased message complexity. Due to the ``redirect" broadcast that every worker agent performs, we have that the message complexity of the PA-CORE protocol is $\Theta (n)$ broadcasts.

\subsection{The PA-CORE-Message-Chain Hybrid Protocol}

The PA-CORE protocol guarantees \emph{Ordered Response} with high probability in sub-linear expected response time, when $\Psi(n) \approx 1$. We know that $\Psi(n) \approx 1$ when the number of participating agents is very large. Consequently, we have not given any guarantees on the performance of the PA-CORE protocol for small $n$.

Similarly to the approach we have presented in the EDD model with a global clock, we propose concurrently running the PA-CORE and Message Chain protocols in EDD model without a global clock. The advantages of doing so are two-fold. Firstly, as before, the expected response time of the protocol will be no worse than that of the Message Chain protocol's, and for large $n$, the expected response time will be sub-linear in $n$. Secondly, our analysis for the expected response time and probability of a correct run are only accurate and informative for large $n$. 

With the addition of a concurrent message chain, the performance of the hybrid protocol is well-known to us for all values of $n$. For small $n$, the message chain guarantees quick \emph{Ordered Response} \emph{w.p.} $1$. For large $n$, the PA-CORE part of the hybrid protocol dominates, and guarantees \emph{Ordered Response} with high probability in sub-linear expected time. As a result, the hybrid PA-CORE-Message-Chain protocol in a setting without a global clock, guarantees \emph{Ordered Response} with high probability in sub-linear expected time.

\section{Discussion and Conclusions}\label{conclusions}

One of the goals of the theory of distributed systems is to understand how various properties of the system affect the solvability and efficiency of distributed problems. To this end, the impact that communicating over channels with probabilistic guarantees on message delivery times has on the efficiency of distributed protocols has received little attention in our community. A comprehensive study of the impact of probabilistic channels on coordination could ultimately involve a taxonomy of different probabilistic assumptions, a variety of coordination problems, tight upper and lower bounds, and an assessment of practical use cases and practical implementation details. This, of course, is a grand research program that goes well beyond the scope of a single conference paper. Our purpose in this paper is to take several initial steps in this direction. To this end, clarity and simplicity are central. Indeed, our goal is to illustrate that such questions can be formulated, and that probabilistic channels can provide a genuine advantage. As we have discussed, asynchronous protocols should be directly implementable in settings in which transmission over a probabilistic channel is guaranteed to succeed eventually with probability 1. In the setting of a natural coordination problem, Ordered Response, and a popular and mathematically accessible exponential distribution over transmission times, we showed 
is possible to improve over the asynchronous solution in a significant manner. This is true both if clocks are synchronized and when they are not. The solutions that we provide are novel, and their analysis is rigorous.  Their simplicity is a feature, and not a bug. 
Being an initial foray into the topic, our treatment opens the door for many extensions, improvements and refinements, left for future research. 
Proving lower bounds and matching upper bounds is one. Exploring other probabilistic distributions, and possibly other coordination problems is another.  Finally, we consider modifying the treatment for concrete practical settings an important open problem. Much is left to explore regarding coordination in this class of models. By establishing that the complexity of asynchronous solutions can be significantly improved on, we have provided strong motivation for such investigations.

\bibliography{bibliography}
\bibliographystyle{plainnat}

\newpage
\appendix

\section{Probability Space Formulation}
\label{prob_space_form}
We define a probability space $\mathcal{S} = (\Omega, \mathcal{F}, \Pr)$  that supports a countably infinite number of exponential random variables. The sample space $\Omega \triangleq \mathbb{R}_{+}^{\mathbb{N}}$ is the space of all non-negative sequences. The $\sigma$-algebra $\mathcal{F} \triangleq \mathbb{B}(\mathbb{R}_{+}^{\mathbb{N}}) = \sigma(\Pi_{i=1}^m (a_i,b_i)\times \Pi_{m+1}^\infty \mathbb{R}_{+} : m\in\mathbb{N}, a_i,b_i\in \mathbb{R}_{+})$ is the Borel sigma-algebra generated by the Tychonoff topology (see \cite{bourbaki2013general}). The probability measure of an event $E = \Pi_{i=1}^m (a_i,b_i)\times \Pi_{m+1}^\infty \mathbb{R}_{+}$, is then defined as $\Pr(E) = \Pi_{i=1}^m \int_{a_i}^{b_i} \lambda e^{-\lambda t}dt$.

We have chosen this probability space due to the continuous nature of time in our model. In simple terms, $E = \Pi_{i=1}^m (a_i,b_i)\times \Pi_{m+1}^\infty \mathbb{R}_{+}$ describes the event that the $i$-th random variable, for $i\leq m$, takes on a value in the interval $(a_i,b_i)$, and that for all $i>m$, it takes on some arbitrary non-negative value. The number $m$ represents the number of messages sent in runs in which the event $E$ occurs. The event $E$ itself is an $m$-dimensional box in a countably infinite dimensional space. $\mathcal{F}$ is the $\sigma$-algebra generated by closure under complement, and under countable unions and intersections of such events $E$ for any $m\in\mathbb{N}$ and $a_i,b_i\in \mathbb{R}_{+}$. Additionally, we have chosen the probability measure $\Pr(E) = \Pi_{i=1}^m \int_{a_i}^{b_i} \lambda e^{-\lambda t}dt$, i.e., the product of $m$ exponential probability measures, since we assume message arrival times are exponential random variables that are statistically independent. 

Let $r \in R(\mathcal{P})$ be a run. We map each run to an element $\omega$ in the sample space $\Omega$. We do this by enumerating the messages sent in $r$ as $m_{i,j}^{(1)}, m_{i,j}^{(2)}, \ldots$ and so forth, for all agent-pairs $i,j \in \mathbb{P}$, where message $m_{i,j}^{(k)}$ is the $k$-th message sent from agent $i$ to agent $j$. Every message $m_{i,j}^{(k)}$ is associated with an exponential random variable $e_{i,j}^{(k)}$ representing its delay. In the run $r$, these random variables take on values $\nu_{i,j}^{(k)}$. Thus, the run $r$ is mapped to an $\omega  \in \Omega$ that is the sequence of $\nu_{i,j}^{(k)}$ for all $i,j \in \mathbb{P}$ and $k\in \mathbb{N}$. Notice that the number of messages in $r$ must be finite. However, the sequence $\omega$ is infinitely long and contains values for message delays that are not sent at all in $r$. So, for the values of $\omega$ that are not associated with a realization of a message delay in $r$, we choose some arbitrary value. For instance, we may choose the value $1$ for all of them, or any other set of non-negative real numbers.

We will mostly use simpler notations whenever the context allows. For example, we define the probability of a run $r \in R(\mathcal{P})$ in which $m$ messages are sent:

\begin{definition}
Fix a run $r \in R(\mathcal{P})$. Let $e_1, e_2, \ldots, e_m$ be the set of \emph{i.i.d.} exponential random variables associated with the messages sent in $r$, let $\nu_1, \nu_2, \ldots, \nu_m$ be their realizations in $r$, and let $f_{e_k}$ be the exponential probability density function (PDF) of $e_k$. Then, the probability of the run $r$ is $\Pr(r) \triangleq \Pi_{k=1}^m f_{e_k}(\nu_k)$.
\end{definition}

\section{Proofs}

\subsection{Proof of Theorem \ref{thm_corr_prob}}\label{thm_corr_prob_proof}

\begin{theorem}
Fix $\D > 0$, and denote $q \triangleq 1-e^{-\lambda \D}$, then the probability of a correct run of the CORE$(\D)$ protocol is $ \sum_{k=0}^{n} \frac{1}{k!}q^{n-k}(1-q)^k$.
\end{theorem}

\begin{proof}
In the CORE$(\D)$ protocol, every agent $i_k$ acts at a time $t_k \geq \D$. Let $e_1, e_2,\ldots,e_n$ be the exponential random variables associated with the messages of the supervisor's broadcast to agents $i_1, i_2, \ldots, i_n$ respectively. for every agent $i_k$, if $e_k \leq \D$, then $t_k = \D$, and otherwise $t_k = e_k$. Hence, for all $k\geq1$:
\begin{equation}
    t_k = \max\{\D, e_k\}
\end{equation}
Notice that the set $\{t_1, t_2, \ldots, t_n\}$ are \emph{i.i.d.} random variables. To calculate the probability $\Pr(t_1 \leq t_2 \leq \cdots \leq t_n)$, we use the Law of Total Probability by conditioning on the values of the set $\{e_k\}_{k=1}^{n}$ relative to $\D$, i.e:
\begin{align}
    E_0 = e_{1} \leq \D, \ e_{2} \leq \D, \ldots, \ &e_{n-2} \leq \D, \ e_{n-1} \leq \D, \ e_{n} \leq \D \\
    E_1 = e_{1} \leq \D, \ e_{2} \leq \D, \ldots, \ &e_{n-2} \leq \D, \ e_{n-1} \leq \D, \ e_{n} > \D \\
    E_2 = e_{1} \leq \D, \ e_{2} \leq \D, \ldots, \ &e_{n-2} \leq \D, \ e_{n-1} > \D, \ e_{n} > \D \\
    &\vdots \\
    E_{n} = e_{1} > \D, \ e_{2} > \D, \ldots, \ &e_{n-2} > \D, \ e_{n-1} > \D, \ e_{n} > \D 
\end{align} 
We use only these $n+1$ combinations out of the possible $2^{n}$, since for any other event the agents will necessarily act out of order, and the probability of Ordered Response conditioned on such an event would be zero. By the Law of Total Probability:
\begin{align}
    \Pr(t_{1} \leq \ldots \leq t_n ) = \sum_{k=0}^{n} \Pr(t_{1} \leq \ldots \leq t_n \mid E_k)\Pr(E_k) \label{tot_prob}
\end{align}
The probability of event $E_k$ is the product of the probability that the first $n-k$ messages are delayed at most $\D$ time, and the probability that the remaining $k$ messages are delayed longer:
\begin{align}
    \Pr(E_k) = q^{n-k}(1-q)^k \label{evnt_prob}
\end{align}
Given an event $E_k$, it holds that the first $n-k$ agents will act simultaneously. However the remaining $k$ messages arrive late, and the agents act immediately upon their delivery. Due to symmetry, for the last $k$ agents, any temporal ordering of their actions is equally likely. Thus:
\begin{align}
    \Pr(t_1 \leq \ldots \leq t_n \mid E_k) =\Pr(t_{n-k+1} \leq \ldots \leq t_n \mid E_k) = \frac{1}{k!} \label{cond_prob}
\end{align}
Plugging \ref{evnt_prob} and \ref{cond_prob} into Equation \ref{tot_prob} yields:
\begin{align}
    \Pr(t_{1} \leq \ldots \leq t_n ) = \sum_{k=0}^{n} \frac{1}{k!}p^{n-k}(1-p)^k
\end{align}
\end{proof}

\subsection{Proof of Lemma \ref{lemma_h_n}}\label{lemma_h_n_proof}

\begin{lemma}
$E[\M] = \frac{H_n}{\lambda}$, where $H_n$ is the $n^{\text{th}}$ harmonic number; i.e. $H_n = \sum_{m=1}^n \frac{1}{m}$.
\end{lemma}

\begin{proof}
Recall that $\M = \max\{e_1, e_2, \ldots, e_n\}$. The random variables $\{e_k\}_{k=1}^n$ are assumed to be \emph{i.i.d.} exponential random variables with parameter $\lambda$. Denote by $T_m$ to be the $m^{\text{th}}$ smallest of the set $\{e_k\}_{k=1}^n$, i.e., $T_1 = \min_k \{e_k\}$ and $T_m = \min \left\{ \{e_k\}_{k=1}^n  \setminus  \{T_1,\ldots,T_{m-1}\} \right\}$. 

It is well known that the minimum of $n$ \emph{i.i.d.} exponential random variables with parameter $\lambda$ is also an exponential random variable with parameter $n\lambda$. Thus, $\E[T_1] = \frac{1}{n\lambda}$. 

Since exponential random variables are continuous, the probability that any two of the $n$ random variables have the same value is 0. Thus, the $n-1$ other random variables all have values strictly larger than $T_1$. However, the memorylessness property of the exponential distribution implies that knowledge of $T_1$ essentially “resets” the values of the other random variables, so that the time between $T_1$ and $T_2$ is the same (distributionally) as the time until the first of $n-1$ \emph{i.i.d.} exponential random variables takes on a value. Hence, $\E[T_2 - T_1] = \frac{1}{(n-1) \lambda}$.

By inductive reasoning, we get that $\forall 1 \leq m \leq n -1: \E[T_{m+1} - T_m] = \frac{1}{(n-m) \lambda}$. As a result, the expected value of the maximum of the $n$ exponential random variables is:
\begin{equation}
     \E[\M] = \E[T_n] = \E\left[ T_1 + \sum_{m=1}^{n-1} (T_{m+1} - T_m) \right] = \sum_{m=0}^{n-1} \frac{1}{(n-m) \lambda} = \sum_{m=1}^{n} \frac{1}{m \lambda} = \frac{H_{n}}{\lambda}
\end{equation}
\end{proof}

\subsection{Proof of Theorem \ref{sync_core_ert}}\label{sync_core_ert_proof}

\begin{theorem}
The expected response time of the CORE$(\D)$ protocol is logarithmic in the number of participating agents i.e., $\E[RT] \in \Oof\ \left(\frac{\log(n)}{\lambda}\right)$.
\end{theorem}

\begin{proof}
From Inequality \ref{ert_bound}:
\begin{equation}
    \E[RT] \leq \D + \E[\M]
\end{equation}
We plug in the expressions for $\D$ and $\E[\M]$ from Corollary \ref{Delta} and Lemma \ref{lemma_h_n}:
\begin{equation}
    \E[RT] \leq \D + \E[\M]\ = \frac{1}{\lambda}\cdot \prs {-\ln(1-\sqrt[n]{p}) + H_n}.
\end{equation}

Notice that $-\ln(1-\sqrt[n]{p}) \in \Theta(\log(n))$ for any $0 < p < 1$:
\begin{align}
    \lim_{n \to \infty} \frac{-ln(1-\sqrt[n]{p})}{\ln(n)} &\underset{\text{Heine}}{=} \lim_{x \to \infty} \frac{-ln(1-\sqrt[x]{p})}{\ln(x)} \underset{\text{L'Hôpital}}{=} \lim_{x \to \infty} \frac{-\frac{1}{1-\sqrt[x]{p}}\cdot \frac{\sqrt[x]{p}\ln(p)}{x^2}}{\frac{1}{x}} \\
    &= \lim_{x \to \infty} \frac{-\frac{\ln(p)}{x}}{\frac{1-\sqrt[x]{p}}{\sqrt[x]{p}}} = \lim_{x \to \infty} \frac{-\frac{\ln(p)}{x}}{\sqrt[-x]{p} - 1} \underset{\text{L'Hôpital}}{=} \lim_{x \to \infty} \frac{\frac{\ln(p)}{x^2}}{\frac{\sqrt[-x]{p}\ln(p)}{x^2}} \\
    &= \lim_{x \to \infty} \sqrt[x]{p} = 1
\end{align}
Additionally, the harmonic numbers roughly approximate the natural logarithm function \cite{conway1998book}, i.e., $H_n \in \Theta(log(n))$. Thus, $\E[RT] \in  \Oof\ \left(\frac{\log(n)}{\lambda}\right)$.
\end{proof}

\subsection{Proof of Theorem \ref{delta_bound_thm}}\label{delta_bound_thm_proof}

\begin{theorem}
Let $T_k = \frac{1}{n}\sum_{m=1}^n \tau_m - \frac{2}{\lambda}$, such that $\forall m : \tau_m \sim Erlang(2,\lambda)$, then there exists a function $f(n)$ such that:
\begin{equation}
   \pr{\max_{1\leq k\leq n} \abs{T_k} \leq \frac{\ln(n)}{\lambda \sqrt{n} }} \geq 1 - f(n)
\end{equation}
and $f(n) \xrightarrow{n \to \infty} 0 $.

\end{theorem}
\begin{proof}

Denote by $Z_k = \frac{1}{n}\sum_{m=1}^n \tau_m - \frac{2}{\lambda}$. Then, $Z_k \sim Erlang(2n, \lambda n)$.
\begin{align}
   &\pr{\max_k \abs{T_k} \leq \frac{\ln(n)}{\lambda \sqrt{n} } } = \pr{\bigcap_{k=1}^n \brs{ \abs{T_k} \leq \frac{\ln(n)}{\lambda \sqrt{n} } }} = 1- \pr{\bigcup_{k=1}^n \brs{ \abs{T_k} \geq \frac{\ln(n)}{\lambda \sqrt{n} } }}\\ 
   &\underset{\text{Union Bound}}{\geq} 1- \sum_{k=1}^n\pr{\abs{T_k} \geq \frac{\ln(n)}{\lambda \sqrt{n} }} = 1- \sum_{k=1}^n\pr{\abs{Z_k - \frac{2}{\lambda}} \geq \frac{\ln(n)}{\lambda \sqrt{n} }} = 
\end{align}
\begin{equation} \label{gen_BND}
    1- n\cdot \pr{\abs{Z_k - \frac{2}{\lambda}} \geq \frac{\ln(n)}{\lambda \sqrt{n} }}
\end{equation}

\begin{align}
    \pr{\abs{Z_k - \frac{2}{\lambda}} \geq \frac{\ln(n)}{\lambda \sqrt{n} }} = &\pr{ \brs{Z_k  \leq \frac{2}{\lambda} - \frac{\ln(n)}{\lambda \sqrt{n} }} \cup \brs{Z_k \geq \frac{2}{\lambda} + \frac{\ln(n)}{\lambda \sqrt{n} }}} = \\
    &\pr{ Z_k  \leq \frac{2}{\lambda} - \frac{\ln(n)}{\lambda \sqrt{n} }} + \pr{Z_k \geq \frac{2}{\lambda} + \frac{\ln(n)}{\lambda \sqrt{n} }}
\end{align}

We upper-bound both probabilities using the Chernoff Bound:

For the left-hand side:
\begin{align}
    &\pr{ Z_k  \leq \frac{2}{\lambda} - \frac{\ln(n)}{\lambda \sqrt{n} }} = \pr{-Z_k \geq \frac{\ln(n)}{\lambda \sqrt{n} } - \frac{2}{\lambda} } \leq \inf_{s>0}\brs{\frac{\E\brk{e^{-sZ_k}}}{e^{s\prs{\frac{\ln(n)}{\lambda \sqrt{n} } - \frac{2}{\lambda}}}}} = \\ \label{chernoff_1}
    &\inf_{s>0}\brs{\prs{1+\frac{s}{n\lambda}}^{-2n}e^{-s\prs{\frac{\ln(n)}{\lambda \sqrt{n} } - \frac{2}{\lambda}}}} \\ 
    & \frac{d}{ds} \brk{\prs{1+\frac{s}{n\lambda}}^{-2n}e^{-s\prs{\frac{\ln(n)}{\lambda \sqrt{n} } - \frac{2}{\lambda}}}} = 0 \\
    &\Rightarrow -\frac{2}{\lambda}\prs{1+\frac{s}{n\lambda}}^{-2n-1}e^{-s\prs{\frac{\ln(n)}{\lambda \sqrt{n} } - \frac{2}{\lambda}}} -\prs{\frac{\ln(n)}{\lambda \sqrt{n} } - \frac{2}{\lambda}}\prs{1+\frac{s}{n\lambda}}^{-2n}e^{-s\prs{\frac{\ln(n)}{\lambda \sqrt{n} } - \frac{2}{\lambda}}} = 0 \\
    &\Rightarrow -\frac{2}{\lambda}\prs{1+\frac{s}{n\lambda}}^{-1} =\frac{\ln(n)}{\lambda \sqrt{n} } - \frac{2}{\lambda} \quad \Rightarrow -\frac{2}{\lambda} =\prs{\frac{\ln(n)}{\lambda \sqrt{n} } - \frac{2}{\lambda}}\prs{1+\frac{s}{n\lambda}}\\
    &\Rightarrow \frac{s}{n\lambda}\prs{\frac{2}{\lambda} - \frac{\ln(n)}{\lambda \sqrt{n}}} = \frac{ln(n)}{\lambda \sqrt{n} } \Rightarrow s = n\lambda \prs{\frac{\frac{\ln(n)}{\lambda \sqrt{n}}}{\frac{2}{\lambda} - \frac{\ln(n)}{\lambda \sqrt{n} }}} = n\lambda \prs{\frac{\ln(n)}{2\sqrt{n} - \ln(n)}} 
\end{align}

Notice that $s>0$ since $2\sqrt{n} > \ln(n)$ for all $n$. Substituting $s$ into line \eqref{chernoff_1} results in:
\begin{align}
    &\label{BND1}\pr{ Z_k  \leq \frac{2}{\lambda} - \frac{\ln(n)}{\lambda \sqrt{n} }} \leq \prs{\frac{2\sqrt{n}}{2\sqrt{n}-\ln(n)}}^{-2n}e^{\sqrt{n}\ln{n}} = \prs{1-\frac{\ln(n)}{2\sqrt{n}}}^{2n}n^{\sqrt{n}}
\end{align}
%\frac{2}{\lambda} + \frac{\ln(n)}{\lambda \sqrt{n}}

For the right-hand side:
\begin{align}
    &\label{chernoff_2} \pr{Z_k \geq \frac{2}{\lambda} + \frac{\ln(n)}{\lambda \sqrt{n} }} \leq \inf_{s>0}\brs{\frac{\E\brk{e^{sZ_k}}}{e^{s\prs{\frac{2}{\lambda} + \frac{\ln(n)}{\lambda \sqrt{n} }}}}} = 
    \inf_{0<s<n\lambda}\brs{\prs{1-\frac{s}{n\lambda}}^{-2n}e^{-s\prs{\frac{2}{\lambda} + \frac{\ln(n)}{\lambda \sqrt{n} }}}} \\ 
    & \frac{d}{ds} \brk{\prs{1-\frac{s}{n\lambda}}^{-2n}e^{-s\prs{\frac{2}{\lambda} + \frac{\ln(n)}{\lambda \sqrt{n} }}}} = 0 \\
    &\Rightarrow \frac{2}{\lambda}\prs{1-\frac{s}{n\lambda}}^{-2n-1}e^{-s\prs{\frac{2}{\lambda} + \frac{\ln(n)}{\lambda \sqrt{n}}}} -\prs{\frac{2}{\lambda} + \frac{\ln(n)}{\lambda \sqrt{n}}}\prs{1-\frac{s}{n\lambda}}^{-2n}e^{-s\prs{\frac{2}{\lambda} + \frac{\ln(n)}{\lambda \sqrt{n}}}} = 0 \\
    &\Rightarrow \frac{2}{\lambda}\prs{1-\frac{s}{n\lambda}}^{-1} =\frac{2}{\lambda} + \frac{\ln(n)}{\lambda \sqrt{n}} \quad \Rightarrow \frac{2}{\lambda} =\prs{\frac{2}{\lambda} + \frac{\ln(n)}{\lambda \sqrt{n}}}\prs{1-\frac{s}{n\lambda}}\\
    &\Rightarrow \frac{s}{n\lambda}\prs{\frac{2}{\lambda} + \frac{\ln(n)}{\lambda \sqrt{n}}} = \frac{ln(n)}{\lambda \sqrt{n} } \Rightarrow s = n\lambda \prs{\frac{\frac{\ln(n)}{\lambda \sqrt{n}}}{\frac{2}{\lambda} + \frac{\ln(n)}{\lambda \sqrt{n} }}} = n\lambda \prs{\frac{\ln(n)}{2\sqrt{n} + \ln(n)}} 
\end{align}

Substituting $s$ into line \eqref{chernoff_2} results in:
\begin{align}
    &\label{BND2}\pr{ Z_k  \geq \frac{2}{\lambda} + \frac{\ln(n)}{\lambda \sqrt{n} }} \leq \prs{\frac{2\sqrt{n}}{2\sqrt{n}+\ln(n)}}^{-2n}e^{-\sqrt{n}\ln{n}} = \prs{1+\frac{\ln(n)}{2\sqrt{n}}}^{2n}n^{-\sqrt{n}}
\end{align}

Substituting both bounds in line \eqref{gen_BND} results in:
\begin{align}
    \pr{\max_k \abs{T_k} \leq \frac{\ln(n)}{\lambda \sqrt{n} } } &\geq 1- n \cdot\brk{\prs{1-\frac{\ln(n)}{2\sqrt{n}}}^{2n}n^{\sqrt{n}} + \prs{1+\frac{\ln(n)}{2\sqrt{n}}}^{2n}n^{-\sqrt{n}}} \\
    &= 1- \brk{\prs{1-\frac{\ln(n)}{2\sqrt{n}}}^{2n}n^{1+\sqrt{n}} + \prs{1+\frac{\ln(n)}{2\sqrt{n}}}^{2n}n^{1-\sqrt{n}}}
\end{align}
Let: 
$$f(n) \triangleq \prs{1-\frac{\ln(n)}{2\sqrt{n}}}^{2n}n^{1+\sqrt{n}} + \prs{1+\frac{\ln(n)}{2\sqrt{n}}}^{2n}n^{1-\sqrt{n}}$$
we now prove that ${f(n) \xrightarrow{n \to \infty} 0 }$:
\begin{equation}
    \lim_{n\rightarrow \infty} f(n) = \lim_{n\rightarrow \infty} \prs{1-\frac{\ln(n)}{2\sqrt{n}}}^{2n}n^{1+\sqrt{n}} + \lim_{n\rightarrow \infty} \prs{1+\frac{\ln(n)}{2\sqrt{n}}}^{2n}n^{1-\sqrt{n}}
\end{equation}
We prove each limit individually:
\begin{align}
    &\lim_{n\rightarrow \infty} \prs{1-\frac{\ln(n)}{2\sqrt{n}}}^{2n}n^{1+\sqrt{n}} 
    = \lim_{n\rightarrow \infty} \prs{1-\frac{\ln(n)}{2\sqrt{n}}}^{\frac{2n\ln(n)}{\ln(n)}}n^{1+\sqrt{n}} = \\
    &\lim_{n\rightarrow \infty} e^{\ln\prs{\prs{1-\frac{\ln(n)}{2\sqrt{n}}}^{\frac{2n\ln(n)}{\ln(n)}}}}n^{1+\sqrt{n}} 
    = \lim_{n\rightarrow \infty} e^{ \frac{2n\ln(n)}{\ln(n)}\ln\prs{1-\frac{\ln(n)}{2\sqrt{n}}} }n^{1+\sqrt{n}} = \\ &\lim_{n\rightarrow \infty} n^{ \frac{2n}{\ln(n)} \ln\prs{1-\frac{\ln(n)}{2\sqrt{n}}} +\sqrt{n} + 1 }
\end{align}
From the Taylor series expansion of $\ln(1-x)$, we have that: $$\forall x\geq0: \ln(1-x) \leq -x -\frac{1}{2}x^2$$ 
Therefore:
\begin{align}
    &\frac{2n}{\ln(n)} \ln\prs{1-\frac{\ln(n)}{2\sqrt{n}}} +\sqrt{n} + 1 
    \leq \frac{{2n}}{{\ln (n)}}\left( { - \frac{{\ln (n)}}{{2\sqrt n }} - \frac{{{{\ln }^2}(n)}}{{8n}}} \right) + \sqrt n  + 1 \\
    &=  - \sqrt n  - \frac{{\ln (n)}}{4} + \sqrt n  + 1 = 1 - \frac{{\ln (n)}}{4} \xrightarrow{n \to \infty} -\infty\\
    &\Rightarrow n^{ \frac{2n}{\ln(n)} \ln\prs{1-\frac{\ln(n)}{2\sqrt{n}}} +\sqrt{n} + 1 } \leq n^{1-\frac{{\ln (n)}}{4}} \\
    &\lim_{n\rightarrow \infty} n^{1-\frac{{\ln (n)}}{4}} = 0 \Rightarrow \lim_{n\rightarrow \infty} \prs{1-\frac{\ln(n)}{2\sqrt{n}}}^{2n}n^{1+\sqrt{n}} =0
\end{align}
Now the other limit:
\begin{align}
    &\lim_{n\rightarrow \infty} \prs{1+\frac{\ln(n)}{2\sqrt{n}}}^{2n}n^{1-\sqrt{n}} 
    = \lim_{n\rightarrow \infty} \prs{1+\frac{\ln(n)}{2\sqrt{n}}}^{\frac{2n\ln(n)}{\ln(n)}}n^{1-\sqrt{n}} = \\
    &\lim_{n\rightarrow \infty} e^{\ln\prs{\prs{1+\frac{\ln(n)}{2\sqrt{n}}}^{\frac{2n\ln(n)}{\ln(n)}}}}n^{1-\sqrt{n}} 
    = \lim_{n\rightarrow \infty} e^{ \frac{2n\ln(n)}{\ln(n)}\ln\prs{1+\frac{\ln(n)}{2\sqrt{n}}} }n^{1-\sqrt{n}} = \\ &\lim_{n\rightarrow \infty} n^{ \frac{2n}{\ln(n)} \ln\prs{1+\frac{\ln(n)}{2\sqrt{n}}} -\sqrt{n} + 1 }
\end{align}
From the Taylor series expansion of $\ln(1-x)$, we have that: \begin{equation}
    \ln(1+x) \leq x -\frac{1}{2}x^2 + \frac{1}{3}x^3
\end{equation}
Therefore:
\begin{align}
    &\frac{2n}{\ln(n)} \ln\prs{1+\frac{\ln(n)}{2\sqrt{n}}} -\sqrt{n} + 1 
    \leq \frac{2n}{\ln(n)} \prs{ \frac{\ln(n)}{2\sqrt{n}} - \frac{\ln^2(n)}{4n} + \frac{\ln^3(n)}{24n\sqrt{n}} } -\sqrt{n} + 1 \\
    &=  \sqrt{n}  - \frac{\ln(n)}{4} + \frac{\ln^2(n)}{12\sqrt{n}} - \sqrt n  + 1 = 1 - \frac{{\ln (n)}}{4} + \frac{\ln^2(n)}{12\sqrt{n}} \xrightarrow{n \to \infty} -\infty \\
    &\Rightarrow n^{ \frac{2n}{\ln(n)} \ln\prs{1+\frac{\ln(n)}{2\sqrt{n}}} -\sqrt{n} + 1 } \leq n^{1 - \frac{{\ln (n)}}{4} + \frac{\ln^2(n)}{12\sqrt{n}}} \\
    &\lim_{n\rightarrow \infty} n^{1 - \frac{{\ln (n)}}{4} + \frac{\ln^2(n)}{12\sqrt{n}}} = 0 \Rightarrow \lim_{n\rightarrow \infty} \prs{1+\frac{\ln(n)}{2\sqrt{n}}}^{2n}n^{1-\sqrt{n}} =0
\end{align}
\end{proof}

\subsection{Proof of Theorem \ref{prob_corr_thm}}\label{prob_corr_thm_proof}

\begin{theorem}
Fix $\D > 0$. The probability of a correct run of the PA-CORE$(\D)$ protocol is at least $\Psi(n)\cdot (1-e^{-0.5\lambda \D})^{(n+1)^2}$.
\end{theorem}

\begin{proof}
By Lemma \ref{prob_corr_given_lemma}, if all messages arrive by time $\D$ and all the agents' local clocks are $\delta$-synchronized, then Ordered Response is assured. Therefore, the probability of a correct run of the protocol is greater or equal to the probability that both of these events occur. 

We now derive the probability that all the messages arrive by time $\D$. For agent $i_k$, let $e_k$ be the exponential random variable associated with the delivery time of the message sent from $i_0$ to $i_k$ in agent $i_0$'s initial broadcast. Let $M_k$ be the maximum of the $n$ random variables that are associated with agent $i_k$'s redirect messages. Then, the probability that all messages arrive by time $\D$ is:
\begin{align}
    \Pr \prs{\max_{1\leq k \leq n+1} \{ e_k + M_k\} \leq \D} &\geq \Pr \prs{\max_{1\leq k \leq n+1} \{ 2\cdot\max\{e_k, M_k\}\} \leq \D} \\
    &= \Pr \prs{\max_{1\leq k \leq n+1} \{ \max\{e_k, M_k\}\} \leq 0.5\D}
\end{align}

where the inequality above follows from the fact that $e_k + M_k \leq 2\cdot \max\{e_k, M_k\}$. Notice that for all $k\geq 1$ the random variables in the set $\left\{\max\{e_k,M_k\}\right\}_{k \geq 1}$ are independent of one another. Hence:
\begin{equation}
    \Pr \prs{\max_{1\leq k \leq n+1} \{ \max\{e_k, M_k\}\} \leq 0.5\D} = \Pr \prs{\{ \max\{e_k, M_k\}\} \leq 0.5\D}^{n+1}
\end{equation}
Notice that $\max\{e_k, M_k\}$ is the maximum over $n+1$ \emph{i.i.d.} exponential random variables, and therefore the CDF of the maximum is the product of their individual CDFs:
\begin{equation}
\Pr \prs{\{ \max\{e_k, M_k\}\} \leq 0.5\D} = (1-e^{-0.5\lambda \D})^{n+1}
\end{equation}
Hence, the probability that all of the messages arrive by time $\D$ is:
\begin{equation}
    \p = (1-e^{-0.5\lambda \D})^{(n+1)^2}
\end{equation}
Let $E_{\D}$ denote the event that all messages arrive by time $\D$, and let $E_\delta$ denote the event that all the agents' local clocks are $\delta$-synchronized. By Lemma \ref{prob_corr_given_lemma} and Theorem \ref{delta_bound_thm}, we obtain:
\begin{equation}
    p \geq \Pr(E_\delta \wedge E_{\D}) = \Pr(E_\delta | E_{\D})\Pr(E_{\D}) \geq \Psi(n)\cdot (1-e^{- 0.5\lambda \D})^{(n+1)^2}
\end{equation}
\end{proof}

\subsection{Proof of Theorem \ref{async_core_ert_thm}}\label{async_core_ert_thm_proof}

We prove a useful Lemma first:
\begin{lemma}\label{async_core_ert_lemma}
The following holds:
$\quad \E\brk{\max_k\brs{e_k, \T^k}} \leq \frac{2H_{n^2+n}}{\lambda}$.
\end{lemma}

\begin{proof}
Recall that $\T^k = \max\brs{\tau_1^k, \tau_2^k, \ldots, \tau_n^k}$. Each timestamp $\tau_m^k$ is measured when the agent $i_k$ receives a ``redirect" message, which completes a message chain composed of two messages. Hence, $\tau_m^k$ is a sum of two exponential random variables, i.e., $\tau_m^k = e_{m,1}^k + e_{m,2}^k$. Notice that $e_{m,1}^k + e_{m,2}^k \leq 2\max\brs{e_{m,1}^k, e_{m,2}^k}$. We have that:
\begin{alignat*}{2}
    &\T^k = \max\brs{\tau_1^k, \tau_2^k, \ldots, \tau_n^k} &&\leq \max_{1\leq m \leq n} \brs{2\max\brs{e_{m,1}^k, e_{m,2}^k}} = 2\max_{\substack{1\leq m \leq n \\ \ell \in \brs{1,2}}} \brs{e_{m,\ell}^k} \\
    &\Rightarrow \quad \E\brk{\max_k\brs{e_k, \T^k}} &&\leq \E\brk{\max_k\brs{e_k, 2\max_{\substack{1\leq m \leq n \\ \ell \in \brs{1,2}}} \brs{e_{m,\ell}^k}}} \\
    & &&\leq \E\brk{\max_k\brs{2e_k, 2\max_{\substack{1\leq m \leq n \\ \ell \in \brs{1,2}}} \brs{e_{m,\ell}^k}}} \\
    & &&= 2\E\brk{\max_k\brs{e_k, \max_{\substack{1\leq m \leq n \\ \ell \in \brs{1,2}}} \brs{e_{m,\ell}^k}}}
\end{alignat*}
Notice that the argument inside the expectation operator is just a maximum over $n^2 + n$ \emph{i.i.d.} exponential random variables. By Lemma \ref{lemma_h_n}, we have that:

\begin{equation}
    \E\brk{\max_k\brs{e_k, \T^k}} \leq 2\E\brk{\max_k\brs{e_k, \max_{\substack{1\leq m \leq n \\ \ell \in \brs{1,2}}} \brs{e_{m,\ell}^k}}} \leq \frac{2H_{n^2+n}}{\lambda}
\end{equation}
\end{proof}

\begin{theorem}
The expected response time of the PA-CORE$(\D)$ protocol is sub-linear in the number of participating agents. In particular, $\E[RT] \in \Oof\prs{\frac{\sqrt{n}\log(n)}{\lambda}}$.
\end{theorem}

\begin{proof}
By Lemma \ref{async_core_ert_lemma}:

\begin{equation}
    \E\brk{RT} \quad \leq \quad \frac{2H_{n^2+n}}{\lambda} + \D + 2\delta n
\end{equation}
We now that prove each expression in the sum is sub-linear in $n$:
\begin{itemize}
    \item For $\frac{2H_{n^2+n}}{\lambda}$:
    \begin{equation}
        H_{n^2+n} \approx \ln(n^2 + n) \in \mathcal{O}(\log(n)) \quad \Rightarrow \quad \frac{2H_{n^2+n}}{\lambda} \in \Oof\prs{\frac{\log(n)}{\lambda}}
    \end{equation}
    \item For $\D$: \\
    By Corollary \ref{delta_tp_cor}:
    \begin{equation}
    \D \approx -\frac{2\ln\prs{1-\sqrt[(n+1)^2]{p}}}{\lambda}\end{equation}
    Notice that $-\ln(1-\sqrt[(n+1)^2]{p}) \in \Theta(\log(n))$ for any $0 < p < 1$:
    \begin{align}
    \lim_{n \to \infty} \frac{-ln(1-\sqrt[(n+1)^2]{p})}{\ln(n+1)} &\underset{\text{Heine}}{=} \lim_{x \to \infty} \frac{-ln(1-\sqrt[(x+1)^2]{p})}{\ln(x+1)} \\
    &\underset{\text{L'Hôpital}}{=} \lim_{x \to \infty} \frac{-\frac{1}{1-\sqrt[(x+1)^2]{p}}\cdot \frac{2 \cdot \sqrt[(x+1)^2]{p}\ln(p)}{(x+1)^3}}{\frac{1}{x+1}} \\
    &= \lim_{x \to \infty} \frac{-\frac{2\ln(p)}{(x+1)^2}}{\frac{1-\sqrt[(x+1)^2]{p}}{\sqrt[(x+1)^2]{p}}} = \lim_{x \to \infty} \frac{-\frac{2\ln(p)}{(x+1)^2}}{\sqrt[-(x+1)^2]{p} - 1} \\
    &\underset{\text{L'Hôpital}}{=} \lim_{x \to \infty} \frac{\frac{4\ln(p)}{(x+1)^3}}{\frac{2 \cdot \sqrt[(x+1)^2]{p}\ln(p)}{(x+1)^3}} = \lim_{x \to \infty} 2\cdot \sqrt[(x+1)^2]{p} = 2
\end{align}
    \item For $2\delta n$: \\
    By definition, $\delta = \frac{\ln(n)}{\lambda \sqrt{n}}$. Hence:
    $$
        2\delta n = \frac{2\sqrt{n}\ln(n)}{\lambda} \in \Oof\prs{\frac{\sqrt{n}\log(n)}{\lambda}}
    $$
\end{itemize}
\end{proof}
\end{document}